\documentclass[twoside,leqno,twocolumn]{article}

\usepackage[letterpaper]{geometry}

\usepackage{ltexpprt}

\usepackage{amssymb, amsmath}
\usepackage{cleveref}
\usepackage{thm-restate}
\usepackage{graphicx,subcaption}
\usepackage{color}

\newcommand{\R}{\mathbb{R}} 
\newcommand{\I}{\mathcal{I}}
\newcommand{\J}{\mathcal{J}}
\newcommand{\Z}{\mathcal{Z}}

\newtheorem{assumption}{Assumption}
\newtheorem{remark}{Remark}

\begin{document}

\title{\Large Fast Convergence of Fictitious Play for Diagonal Payoff Matrices}
\author{Jacob Abernethy\thanks{School of Computer Science, Georgia Institute of Technology. Email: \texttt{prof@gatech.edu}}
\and Kevin A.~Lai\thanks{School of Computer Science, Georgia Institute of Technology. Email: \texttt{kevinlai@gatech.edu}}
\and Andre Wibisono\thanks{School of Computer Science, Georgia Institute of Technology.
%Department of Computer Science, Yale University. 
Email: \texttt{wibisono@gatech.edu}}
}

\date{}

\maketitle

% Copyright Statement
% When submitting your final paper to a SIAM proceedings, it is requested that you include 
% the appropriate copyright in the footer of the paper.  The copyright added should be 
% consistent with the copyright selected on the copyright form submitted with the paper.
% Please note that "20XX" should be changed to the year of the meeting.

% Default Copyright Statement
\fancyfoot[R]{\scriptsize{Copyright \textcopyright\ 2021 by SIAM\\
Unauthorized reproduction of this article is prohibited}}

% Depending on which copyright you agree to when you sign the copyright form, the copyright 
% can be changed to one of the following after commenting out the default copyright statement
% above.

%\fancyfoot[R]{\scriptsize{Copyright \textcopyright\ 20XX\\
%Copyright for this paper is retained by authors}}

%\fancyfoot[R]{\scriptsize{Copyright \textcopyright\ 20XX\\
%Copyright retained by principal author's organization}}

%\pagenumbering{arabic}
%\setcounter{page}{1}%Leave this line commented out.

\begin{abstract} \small\baselineskip=9pt 
Fictitious Play (FP) is a simple and natural dynamic for repeated play in zero-sum games. Proposed by Brown in 1949, FP was shown to converge to a Nash Equilibrium by Robinson in 1951, albeit at a slow rate that may depend on the dimension of the problem. 
In 1959, Karlin conjectured that FP converges at the more natural rate of $O(1/\sqrt{t})$. However, Daskalakis and Pan disproved a version of this conjecture in 2014, showing that a slow rate can occur, although their result relies on adversarial tie-breaking. In this paper, we show that Karlin's conjecture is indeed correct for the class of diagonal payoff matrices, as long as ties are broken lexicographically. Specifically, we show that FP converges at a $O(1/\sqrt{t})$ rate in the case when the payoff matrix is diagonal.
We also prove this bound is tight by showing a matching lower bound in the identity payoff case under the lexicographic tie-breaking assumption.
\end{abstract}

\section{Introduction}
 In a two-player zero-sum game, we are given a payoff matrix $A \in \R^{n\times m}$, whose $ij$-th entry denotes how much the row player pays the column player when the two players play actions $i$ and $j$ respectively. When each player selects their actions randomly, with the row player sampling $i$ from some distribution $x \in \Delta_n$ and the column player sampling $j$ from some distribution $y \in \Delta_m$ (where $\Delta_d$ is the $(d-1)$-dimensional simplex in $\R^d$), the expected gain to the column player (or equivalently, the expected loss to the row player) is exactly $x^\top A y$. Following the work of von Neumann and Nash, we say that a pair of distributions $(x^*,y^*)$ is at a \textit{minimax point} or \emph{Nash Equilibrium} if we have:
\begin{align*}
(x^\ast)^\top A y \le (x^\ast)^\top A y^\ast \le x^\top A y^\ast  % ~~~~ \text{ for all } ~  \, (x,y) \in  \Delta_n\times \Delta_n
\end{align*}
for all $(x,y) \in  \Delta_n\times \Delta_n$.
 
In what might be considered the fundamental theorem of game theory, von Neumann proved \cite{neumann1928theorie} that every zero-sum game admits an equilibrium pair; Nash later showed the same holds for non-zero-sum games \cite{nash1951non}. Von Neumann's theorem is often stated in terms of the equivalence of a min-max versus a max-min:
\begin{align*}
\min_{x\in \Delta_n} \max_{y\in \Delta_m} x^\top A y = \max_{y\in \Delta_m}\min_{x\in \Delta_n} x^\top A y.
\end{align*}
It is easy to check that the minimizer of the left hand side and the maximizer of the right exhibit the desired equilibrium pair.

One of the earliest methods for computing Nash Equilibria in zero sum games is \emph{fictitious play} (FP), proposed by Brown~\cite{brown1949some,brown1951iterative}. FP is perhaps the simplest dynamic one might envision for repeated play in a game---in each round, each player considers the empirical distribution of the actions of the other player and selects their action as the best response to this statistic. Formally, we can define state variables $x^{(t)}, y^{(t)}$ at each iteration $t$ and update according to the rule
\begin{align}\label{eq:fpsimple}
\begin{split}
x^{(t+1)} &= x^{(t)} + e_{ i^{(t)} } , \;\; i^{(t)}=\arg\min_{i \in [n]} \, e_i^\top Ay^{(t)} \\
y^{(t+1)} &= y^{(t)} + e_{ j^{(t)} } , \;\; j^{(t)} =\arg\max_{j \in [m]} \, (x^{(t)})^\top Ae_j.
\end{split}
\end{align}
where $e_\ell$ is the $\ell^{th}$ standard unit basis vector. Brown conjectured that the scaled state variables $(\hat x^{(t)}, \hat y^{(t)}) = (\frac{1}{t} x^{(t)},\frac{1}{t} y^{(t)})$ would converge to a minimax point, and in 1951, Robinson showed that FP converges 
asymptotically to a minimax point \cite{robinson1951iterative}.

In addition to having an intuitive game theoretic interpretation, FP has several other strengths. The update itself is simple, having no step-size parameter to tune. Since the row and column players only need $O(n)$ or $O(m)$ memory respectively, the dynamic is also amenable to distributed computation. These qualities have have FP an appealing object of study, and many works have sought to prove convergence of FP in more general settings \cite{miyasawa1961convergence, shapley1964some, MS96, brandt2010rate}. FP has also inspired many algorithms, such as the Follow-The-Perturbed Leader algorithm \cite{HS02} and other online learning algorithms. More recently, DeepMind used an algorithm called Prioritized Fictitious Self Play as part of the training for their AlphaStar program for playing competitive Starcraft \cite{vinyals2019grandmaster}.

Despite the extensive work on FP, much of it has focused on asymptotic convergence, leaving significant questions about the convergence rate of the dynamic. While not stated as such, Robinson's 1951 proof actually implies that FP converges to within $O(t^{-\frac{1}{m+n-2}})$ of the equilibrium pair $(x^*, y^*)$ after $t$ rounds of play \cite{robinson1951iterative}.  Robinson's result utilized a recursive argument that successively eliminates actions of the players, and she did not address whether this was a tight rate. In what is often known as \emph{Karlin's Conjecture} from 1959, Karlin~\cite{karlin1959games} suggested that the true rate may be significantly faster, perhaps on the order of $O(t^{-1/2})$. This remained an open question for decades, but was seemingly put to rest in 2014 by Daskalakis and Pan~\cite{daskalakis2014counter} who were able to produce an instance of a game and a FP dynamic for which the convergence rate was $\Omega(t^{-1/n})$, in particular the rate is slow and depends on the number of actions, similar to the bound of Robinson. Their lower bound construction follows along the same lines as the upper bound of Robinson, recursively generating harder instances as more actions are given to the players.

The goal of our work is to show that Karlin's conjecture may have only been \textit{ostensibly} resolved, and we argue that a slightly more precise version of the conjecture is likely to be true, namely that a particular form of FP will admit a rate of $O(t^{-1/2})$. The imprecise aspect of Karlin's conjecture is that the $\arg\min$ and $\arg\max$ in \eqref{eq:fpsimple} are not \emph{well-defined} to the extent that many solutions can exist in the event of ties. Daskalakis and Pan distinguish between the model in which ties arising in \eqref{eq:fpsimple} can be broken in an arbitrary (adversarial) fashion and the model in which ties are broken lexicographically; they acknowledge that their lower bound holds only in the former case. Their lower bound construction heavily exploits the ill-defined nature of \eqref{eq:fpsimple}, employing carefully-constructed tie-making and adversarial tie-breaking to obtain the slow rate. We emphasize that one of the appealing properties of FP is that it is a natural game dynamic, yet the dynamic proposed by Daskalakis and Pan, while technically satisfying a definition of fictitious play, is by no means natural.

We consider the convergence of a well-defined version of FP with \emph{lexicographic tie-breaking}, where the $\arg\min$ and $\arg\max$ functions break ties by selecting the winner with the smallest index. Lexicographic tie-breaking is one of the simplest tie-breaking methods, being the default when writing a min or max in code. From a game theoretic perspective, lexicographic tie-breaking corresponds to each player having fixed tie-breaking preferences between actions. We show that this version of FP has a rate of $O(t^{-1/2})$ for the class of diagonal payoff matrices, which includes the matrix used in the lower bound of Daskalakis and Pan. This is the first positive result showing any improvement over Robinson's result for matrices of size $3\times 3$ or larger. We further provide a lower bound of $\Omega(t^{-1/2})$ under lexicographic tie-breaking, showing that our iteration complexity bound is tight.

Our analysis gives a tight characterization of the how the FP dynamic evolves in the diagonal case. We show how lexicographic tie-breaking causes the dynamic to behave in a specific way, which allows us to prove our upper and lower bounds. To our knowledge, this is the first such work that leverages lexicographic tie-breaking to prove fast convergence, and we hope our work lays the groundwork for proving the $O(t^{-1/2})$ upper bound for arbitrary payoff matrices. We conclude by discussing some ways to extend our analysis and related open questions.

\subsection{Related work}

We give a brief overview of prior work on fictitious play and related game dynamics.

\paragraph{Fictitious Play} The original formulation of FP was by Brown~\cite{brown1949some,brown1951iterative}, where he mentions both discrete and continuous time dynamics. Since then, FP has been studied extensively--many works have explored the asymptotic convergence of fictitious play in various game settings \cite{miyasawa1961convergence, shapley1964some, MS96, brandt2010rate}, while another notable line of work examines properties of the continuous time version of FP \cite{harris1998rate, OS11, OS14, SK17}. The first convergence rate for FP was shown by Robinson~\cite{robinson1951iterative}, who proves that FP achieves a rate of $O(t^{-\frac{1}{m+n-2}})$ under arbitrary tie-breaking. Karlin~\cite{karlin1959games} later conjectures that the convergence rate is $O(t^{-\frac{1}{2}})$. This matches the convergence rate of several related dynamics based on no-regret algorithms, as described below. Moreover, FP appears to always achieve this rate empirically, as we illustrate in \Cref{fig:empirical_karlin}. Daskalakis and Pan~\cite{daskalakis2014counter} construct a counter-example for Karlin's strong conjecture using carefully designed adversarial tie-breaking rules, showing that FP for a zero-sum game on the $n\times n$ identity matrix has a worst-case convergence rate of $\Omega(t^{-\frac{1}{n}})$.

\paragraph{No-regret dynamics} The literature on so-called \emph{online learning} \cite{cesa2006prediction} considers the family of problems in which an algorithm must make a decision on each of a sequence of $T$ rounds---this could be a discrete choice among $n$ alternatives, for example, or a real-valued parameter vector $\theta$---and then the decision is evaluated according to some loss which provides appropriate feedback. The algorithm's long-term goal is to minimize its \emph{regret}, defined as the difference between the algorithm's cumulative loss and the loss of the best fixed action in hindsight. There has been a great deal of work on developing such \emph{no regret algorithms} \cite{hazan2016introduction,cesa2006prediction}, where the overall regret scales sublinearly with $T$, and such algorithms are often used to design game dynamics in order to ensure equilibrium-convergent behavior; see for example \cite{abernethy2018faster,BP19b}. Regret-minimizing algorithms are convenient choices for producing well-behaved dynamics in zero-sum games because the rate of convergence to equilibria can be established directly through the time-averaged regret of the players' actions.

What makes understanding Fictitious Play so challenging is that the players' actions in this dynamic do not necessarily exhibit vanishing regret and hence we can not immediately use such results to reason about convergence. But the procedures used by fictitious players \emph{resemble} those of no-regret algorithms in the following sense. FP can be viewed as having both players update their actions using the ``Follow-The-Leader'' (FTL) algorithm, a strategy described by \cite{kalai2005efficient} who provide a precise example showing that FTL can have linear regret. But the FTL algorithm motivates their more sophisticated algorithm,  Follow-The-Perturbed-Leader (FTPL), which applies the follow-the-leader rule only after some random noise is added to the total losses, and they show this has the desired no-regret property. A related algorithm, Follow-The-Regularized-Leader (FTRL) \cite{HS02,hazan2016introduction}, can also guarantee no regret\footnote{There is a surprising and interesting connection between FTPL and FTRL; we refer the reader to \cite{abernethy20178,abernethy2014online} for more.}. FTPL and FTRL both have a bound of $O(t^{-\frac{1}{2}})$ on their time-averaged regret, and it can be shown that this implies a $O(t^{-\frac{1}{2}})$ for the corresponding dynamics in a zero-sum game if both players utilize such algorithms.

Since FTL is provably not a no-regret algorithm, as mentioned above, we must develop novel techniques beyond the no-regret framework to reason about convergence to equilibrium. To summarize the challenges here: while FTL does not exhibit vanishing regret on particular sequences of losses, in the FP dynamic each player's observed loss sequence is generated from the FTL actions of the other player, and such sequences may be ``easier'' and do not induce large regret. We also note that there are benefits to the FP dynamic beyond its naturalness and simplicity: algorithms like FTRL and FTPL require a time-dependent step-size, whereas FTL has no notion of step-size. In the event of a fixed step-size, while it remains unknown whether FTRL- or FTPL-based dynamics converge in general, there have been a number of works on this topic in specific settings \cite{palaiopanos2017multiplicative,daskalakis2018last,BP19b}.

\section{Preliminaries}

\paragraph{Notation.}
Let $\R$ denote the set of real numbers. 
Let $[n] = \{1,\dots,n\}$ denote the set of actions.
Let $\Delta_n = \{x \in \R^n \colon x_i \ge 0, \sum_{i=1}^n x_i = 1\}$ denote the probability simplex.
For $i \in [n]$, let $e_i$ denote the $i^{\text{th}}$ elementary basis vector.
Let $I_n$ denote the $n\times n$ identity matrix. 
For a vector $v \in \R^n$, we let $v_i \in \R$ denote its $i^{\text{th}}$ entry.

\subsection{The minimax theorem and the duality gap} 

In this paper, we assume we are working with a square payoff matrix $A \in \R^{n\times n}$.
The decision set for both row and column players is the simplex $\Delta_n$, which is the set of probability distributions over the actions $[n]$.

A {\em minimax point} is a point $(x^\ast, y^\ast) \in \Delta_n \times \Delta_n$ which satisfies:
\begin{align}\label{Eq:Minimax}
(x^\ast)^\top A y \le (x^\ast)^\top A y^\ast \le x^\top A y^\ast
\end{align}
for all $(x,y) \in  \Delta_n\times \Delta_n$.
By Von Neumann's minimax theorem, we know that a minimax point exists for any $A$, although it is not necessarily unique in general. 

For any $x, y \in \R^n$ we define the {\em duality gap} $\psi \colon \R^n \times \R^n \to \R$ as
\begin{align}\label{Eq:DG}
\psi(x,y) = \max_{\tilde y \in \Delta_n} \, x^\top A \tilde y - \min_{\tilde x \in \Delta_n} \, \tilde x^\top Ay.
\end{align}

While $\psi$ is defined on all of $\R^n \times \R^n$, we are interested in its behavior on $\Delta_n \times \Delta_n$.
It holds that $\psi(x,y) \geq 0$ for any $(x,y) \in \Delta_n \times \Delta_n$. Furthermore, we can characterize a minimax point as the minimizer of the duality gap (see Appendix~\ref{Sec:MinimaxProof} for a proof).
\begin{lemma}\label{Lem:Minimax}
	A point $(x^\ast,y^\ast) \in \Delta_n \times \Delta_n$ is a minimax point if and only if $\psi(x^\ast,y^\ast) = 0$.
\end{lemma}
\noindent In this sense, $\psi(x,y)$ is a measure of distance of some $(x,y) \in \Delta_n \times \Delta_n$ to the equilibrium $(x^\ast,y^\ast)$.

We note that in the diagonal case when $A$ is a diagonal matrix with positive diagonal entries $A_{ii} > 0$, 
the minimax point $(x^\ast,y^\ast)$ is unique with $x^\ast_i = y^\ast_i \propto 1/A_{ii}$ for all $i \in [n]$.

\subsection{Fictitious Play}

In this work when we use the term \emph{dynamic} we are referring to the sequence of actions generated by two players in a repeated zero-sum game when each uses a particular decision rule to select their moves. The {\em Fictitious Play} (FP) dynamic arises when each player's decision rule is as follows: compute the empirical distribution of the previous actions taken by their opponent, and then choose \emph{the} action on this round that is a best response to this empirical distribution.
(We said ``the action'' rather than ``an action'' since we must specify a unique choice even when there are ties; see below.)

Concretely, the Fictitious Play algorithm starts at any $(x^{(1)}, y^{(1)}) \in \Delta_n \times \Delta_n$\footnote{Classically, FP is initialized at $(x^{(0)}, y^{(0)}) = (\mathbf{0}, \mathbf{0})$ at time $t = 0$, so $(x^{(1)}, y^{(1)})$ is in $\Delta_n \times \Delta_n$.}
and performs the following update at each time $t \ge 1$:
\begin{align}\label{Eq:FP}
\begin{split}
x^{(t+1)} &= x^{(t)} + e_{i^{(t)}} ,\;\; i^{(t)}=\arg\min_{i \in [n]} \, e_i^\top Ay^{(t)} \\
y^{(t+1)} &= y^{(t)} + e_{j^{(t)}} ,\;\; j^{(t)}=\arg\max_{j \in [m]} \, (x^{(t)})^\top Ae_j.
\end{split}
\end{align}
In order to be well-defined, we need to specify how to break possible ties in the $\arg\min$ and $\arg\max$.
We do this via a lexicographic ordering; see Section~\ref{Sec:TieBreak} for detail.

\begin{assumption}\label{as:no_tie}
Ties in the FP dynamic~\eqref{Eq:FP} are broken according to a lexicographic order.
\end{assumption}

We refer to the iterates $(x^{(t)}, y^{(t)}) \in \R^n \times \R^n$, $t \ge 1$, as the dynamic of the FP algorithm.
We also refer to the time $t \ge 1$ as the {\em rounds} of the algorithm.

At each time $t \ge 1$ we can consider the scaled iterates, which always lie on the simplex:
\begin{align}
\left(\hat x^{(t)}, \hat y^{(t)}\right) = \left( \frac{x^{(t)}}{t}, \, \frac{y^{(t)}}{t} \right) \in \Delta_n \times \Delta_n.
\end{align}
The main focus of this work is to understand how quickly $(\hat x^{(t)}, \hat y^{(t)})$ converges to $(x^\ast, y^\ast)$ as $t \to \infty$.

In particular, we can measure the speed of convergence via how fast the duality gap $\psi(\hat x^{(t)}, \hat y^{(t)})$ converges to $0$.
The classical result by Robinson~\cite{robinson1951iterative} shows that for any payoff matrix $A \in \R^{n \times n}$, $\psi(\hat x^{(t)}, \hat y^{(t)}) = O(t^{-\frac{1}{2n-2}})$.
Karlin's conjecture~\cite{karlin1959games} asks whether Fictitious Play in fact has a faster rate $\psi(\hat x^{(t)}, \hat y^{(t)}) = O(t^{-\frac{1}{2}})$.
Daskalakis and Pan~\cite{daskalakis2014counter} disprove the {\em strong} formulation of Karlin's conjecture by showing that even in the identity case ($A = I_n$), we can have $\psi(\hat x^{(t)}, \hat y^{(t)}) = \Omega(t^{-\frac{1}{n}})$ if {\em arbitrary} tie-breaking is allowed.
Our results in this paper prove the {\em weak} formulation of Karlin's conjecture in the diagonal case by showing that indeed $\psi(\hat x^{(t)}, \hat y^{(t)}) = \Theta(t^{-\frac{1}{2}})$ if we use a {\em fixed} (lexicographic) tie-breaking, as in~\Cref{as:no_tie}.

We finish this introduction by noting that we can evaluate $\psi$ on either $(\hat x^{(t)}, \hat y^{(t)})$ or $(x^{(t)}, y^{(t)})$, and they are related by: $\psi(x^{(t)}, y^{(t)}) = t \, \psi(\hat x^{(t)} \hat y^{(t)})$.
A basic fact about Fictitious Play is that $\psi(x^{(t)}, y^{(t)})$ is always non-decreasing (see Section~\ref{Sec:PsiIncreasing}; see also Appendix~\ref{App:Geom} for a geometric view).

\section{Main Results}

We give an overview of the results in this paper.
We provide details in Section~\ref{Sec:Diagonal}.

\subsection{Fast convergence of Fictitious Play in diagonal case}

Our first main result is to show that Karlin's conjecture is indeed true for the class of diagonal payoff matrices, as long as the tie-breaking \Cref{as:no_tie} holds true. 

We assume we are in the {\em diagonal case}, namely when $A \in \R^{n \times n}$ is a diagonal matrix with positive diagonal entries $A_{ii} > 0$.
This is an important special case, as it includes the identity case used in the lower bound by Daskalakis and Pan~\cite{daskalakis2014counter}. 
This shows that the slow-converging construction in~\cite{daskalakis2014counter} is prohibited under \Cref{as:no_tie}.

\begin{restatable}{theorem}{ThmUpper}\label{Thm:Upper}
Assume $A$ is a diagonal matrix with $0 < A_{ii} \le A_{\max}$ for all $i \in [n]$.
Under~\Cref{as:no_tie}, for any $(x^{(1)}, y^{(1)}) \in \Delta_n \times \Delta_n$, the FP dynamic~\eqref{Eq:FP} satisfies for all $t \ge 1$
\begin{align*}
\psi(\hat x^{(t)}, \hat y^{(t)}) = O \left(\frac{A_{\max}}{\sqrt{t}}\right).
\end{align*}
\end{restatable}

We provide a proof of Theorem~\ref{Thm:Upper} in Section~\ref{Sec:UpperProof}, relying on three main properties.
First, we show that under \Cref{as:no_tie} the FP dynamic in the diagonal case alternates between two phases, which we call \emph{sync} and \emph{split} phases. We use the term \emph{sync-split pair} to denote a pair of consecutive phases consisting of a sync phase followed by a split phase. 
We note that for this special structure to hold, we crucially require the lexicographic tie-breaking condition.
Second, we show that the duality gap (on the unscaled iterate) can only increase by a constant over the course of a sync-split pair. 
Finally, we show that the duration of each sync-split pair is at least the value of the duality gap at the start of the sync-split pair.
From these properties, we can derive the rate. 

To get some intuition, we can consider the case when 1) the duality gap always increases by a constant $\epsilon > 0$ during each sync-split pair and 2) the duality gap at the start of the first sync-split pair is 0. Then the duality gap at the start of the $s^{\text{th}}$ sync-split pair is $(s-1)\epsilon$.  %Then after $s^{\tetxt{th}}$ sync-split pair, the total duality gap will be $tc$. 
Meanwhile, the number of rounds required to complete these $s$ sync-split pairs is $\sum_{r=1}^{s-1} (r-1)\epsilon = \Theta(s^2 \epsilon)$ because the duration of each sync-split pair is at least the duality gap at the start of the pair. We can see that the duality gap grows as the square root of the number of rounds.

This result proves Karlin's conjecture under the tie-breaking \Cref{as:no_tie}.
Previously, the FP dynamic was known to have $O(t^{-1/2})$ convergence rate only in the case $n=2$ (in which case the lower bound from~\cite{daskalakis2014counter} is also $\Omega(t^{-1/2})$).
Therefore, our result greatly expands the class of games for which the FP dynamic has been shown to converge quickly to equilibrium. 
We note that the upper bound above is independent of the dimension $n$. 

This result confirms that Fictitious Play has a fast $O(t^{-1/2})$ convergence in the diagonal case, achieving a rate similar to no-regret algorithms from online learning despite the fact that FP is not a no-regret algorithm. This result also raises further questions and potential generalizations, which we describe in Section~\ref{Sec:Disc}.

\subsection{Matching lower bound in the identity case}

Our second result is to show a lower bound in the identity case, namely when $A = I_n$
and Fictitious Play is initialized at the vertices of the simplex.

\begin{restatable}{theorem}{ThmLower}\label{Thm:Lower}
	Assume $A = I_n$
	and $(x^{(1)},y^{(1)}) = (e_i,e_j)$ for some $i,j \in [n]$.
	Under~\Cref{as:no_tie}, the FP dynamic~\eqref{Eq:FP} satisfies for all $t \ge 1$:
	\begin{align*}
	\psi(\hat x^{(t)}, \hat y^{(t)}) = \Omega\left(\frac{1}{n\sqrt{t}}\right).
	\end{align*}
\end{restatable}

We provide a proof of Theorem~\ref{Thm:Lower} in Section~\ref{Sec:LowerProof}.
Our proof relies on three main properties.
First we note that since the duality gap is integral for $A=I_n$; this means that if the duality gap increases on some round then it must increase by at least $1$.
Second, we show that within any sequence of $n$ sync-split pairs, the duality gap must increase at least once.
Third, we show that the length of a sync-split pair is proportional to the duality gap. These properties imply the desired lower bound.

We note this lower bound has the same dependence on $t$ as the upper bound in Theorem~\ref{Thm:Upper}.
Thus, in the identity case, the Fictitious Play dynamic indeed has $\psi(\hat x^{(t)}, \hat y^{(t)}) = \Theta(t^{-1/2})$.
We also note this identity case is the same setting used by Daskalakis and Pan~\cite{daskalakis2014counter}, who show that under the arbitrary tie-breaking condition, the Fictitious Play dynamic slows down to $\psi(\hat x^{(t)}, \hat y^{(t)}) = \Omega(t^{-1/n})$.

While analogous lower bounds exist for the case $n=2$~\cite{daskalakis2014counter},
to our knowledge, this result is the first lower bound for general dimension $n$.
However, we note the lower bound has a dependence on $n$, which is likely suboptimal.
We leave improving that dependence to future work.

\section{Analysis of Fictitious Play}
\label{Sec:Diagonal}

In this section we present the analysis of Fictitious Play in the diagonal case.
We begin by defining the lexicographic tie-breaking condition in Section~\ref{Sec:TieBreak}.
In Section~\ref{Sec:SyncSplit} we describe the structure of the Fictitious Play dynamic as a sequence of sync and split phases.
We provide the arguments for the upper bound in Section~\ref{Sec:UpperProof}
and for the lower bound in Section~\ref{Sec:LowerProof}.
We defer proofs and further details to Sections~\ref{Sec:DetailsUpper} and~\ref{Sec:DetailsLower}.

\subsection{Tie-breaking condition}
\label{Sec:TieBreak}

For the Fictitious Play algorithm to be well-defined, we need to specify how to choose $i^{(t)}$, $j^{(t)}$ when the $\arg\min$ or $\arg\max$ in~\eqref{Eq:FP} is not unique.
We assume we break ties following a {\em lexicographic} tie-breaking order, which is defined via two arbitrary but fixed permutations, $\sigma_x, \sigma_y \colon [n] \to [n]$.
Concretely, at time $t \ge 1$ we define the set of (possibly non-unique) best-responses as
\begin{align*}
\I^{(t)} &:= \arg\min_{i \in [n]} \, e_i^\top Ay^{(t)}  \\
&\,\,= \textstyle \left\{i \in [n] \colon \, e_i^\top Ay^{(t)} = \min_{h \in [n]} e_h^\top A y^{(t)} \right\}, \\
\J^{(t)} &:= \arg\max_{j \in [n]} \, (x^{(t)})^\top Ae_j \\
&\,\,= \textstyle \left\{j \in [n] \colon \, (x^{(t)})^\top Ae_j = \max_{h \in [n]} (x^{(t)})^\top Ae_h \right\}.
\end{align*}
We may now precisely specify the Fictitious Play dynamic~\eqref{Eq:FP}, under~\Cref{as:no_tie}, as
\begin{eqnarray*}
	i^{(t)} := \arg\min_{i \in \I^{(t)}} \sigma_x(i) 
	& \text{ and } &
	j^{(t)} := \arg\min_{j \in \J^{(t)}} \sigma_y(j).
\end{eqnarray*}
These $\arg\min$'s are guaranteed to be unique since $\sigma_x, \sigma_y$ are permutations.

This lexicographic tie-breaking condition is stronger than the arbitrary (adversarial) tie-breaking condition in Daskalakis and Pan~\cite{daskalakis2014counter}, which corresponds to allowing the permutations $\sigma_x, \sigma_y$ to change with time $t$.
Therefore, this lexicographic tie-breaking condition imposes more structure on the Fictitious Play dynamic, which allows us to escape the lower bound construction of~\cite{daskalakis2014counter}.
Indeed, under Assumption~\ref{as:no_tie}, we can characterize the FP dynamic in the diagonal case as an alternating sequence of phases, which allows us to prove Karlin's conjecture.

\subsection{Sync-split pairs}
\label{Sec:SyncSplit}

A key property that the lexicographic tie breaking implies is that the Fictitious Play dynamic goes through a sequence of sync and split phases.
Let us define some terms to help our discussion.

\begin{Definition}[Sync and split rounds]\label{def:round}
At round $t \ge 1$, suppose $x$ plays action $i^{(t)} = i$ and $y$ plays action $j^{(t)} = j$.
Then we say the {\em type} of round $t$ is $(i,j)$. Moreover:

\begin{enumerate}
\item If $i = j$, then we say round $t$ is a {\em sync round} (in particular, a {\em sync$(i,i)$ round}).
\item If $i \neq j$, then we say round $t$ is a {\em split round} (in particular, a {\em split$(i,j)$ round}).
\end{enumerate}

\end{Definition}

\noindent Along the Fictitious Play dynamic, the duality gap only increases when the type changes. It is for this reason that we partition the rounds into phases of the same type.

\begin{Definition}[Phase]\label{def:phase}
A {\em phase} is a maximal block of consecutive rounds of the same type.
That is, rounds $\{t,\dots,t+s-1\}$ form a phase if
\begin{itemize}
	\item{they all have the same type $(i,j)$;}
	\item{round $t+s$ is not type $(i,j)$; and }
	\item{round $t-1$ (if $t \ge 2$) is not type $(i,j)$.}
\end{itemize}
We also define the following:

\begin{enumerate}
\item If $i = j$, then we say this phase is a {\em sync phase} (in particular, a {\em sync$(i,i)$ phase}).
\item If $i \neq j$, then we say this phase is a {\em split phase} (in particular, a {\em split$(i,j)$ phase}).
\end{enumerate}
\end{Definition}

\begin{figure}[h!]
	\centering
	\includegraphics[width=.48\textwidth]{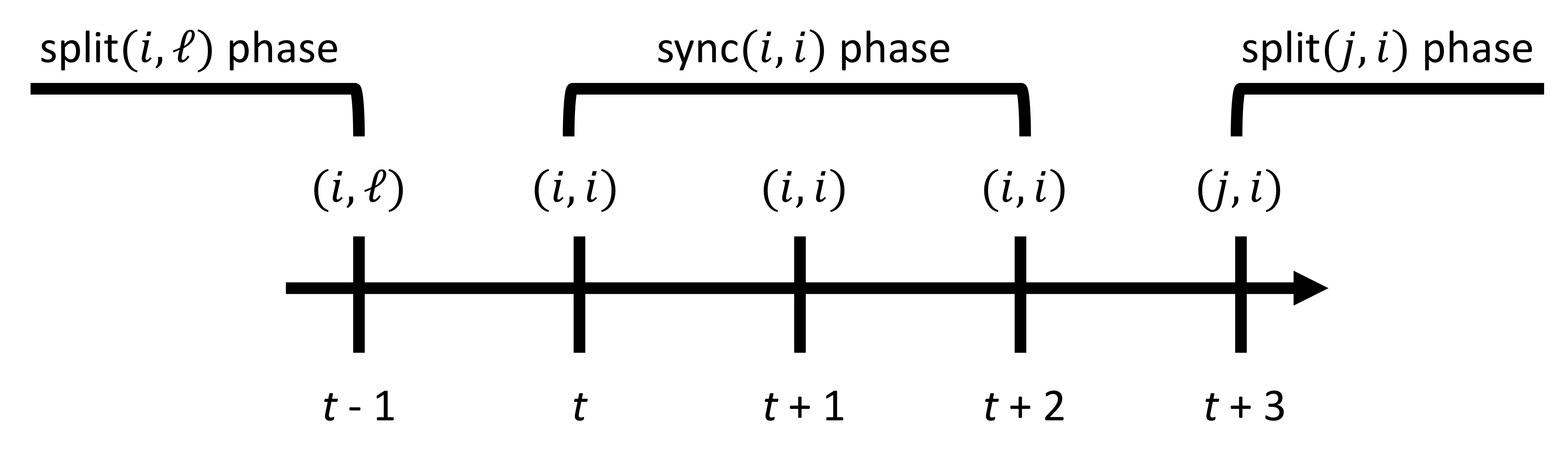}
	\caption{Illustration of Definitions~\ref{def:round} and~\ref{def:phase}. Each tick on the line represents a round. The round number is below the tick and the round type is above the tick (e.g.\ round $t-1$ is an $(i,\ell)$ round). Rounds $\{t, \, t+1, \, t+2\}$ form a sync$(i, i)$ phase.\label{fig:phase}}
\end{figure}

A nice property is that the Fictitious Play dynamic alternates through sync and split phases, as we show in \Cref{Lem:SyncSplit}. This motivates the following definition:

\begin{Definition}[Sync-split pair]
A {\em sync-split pair} (in particular, a {\em sync-split$(i \to j)$ pair}) is a pair of phases consisting of a sync$(i,i)$ phase followed by a split$(j,i)$ phase for some $j \neq i$.
\end{Definition}

\begin{restatable}{lemma}{LemSyncSplit}\label{Lem:SyncSplit}
Assume $A$ is a diagonal matrix.
Under~\Cref{as:no_tie}, the FP dynamic~\eqref{Eq:FP} proceeds through a sequence of sync and split phases:
\begin{enumerate}
  \item A sync$(i,i)$ phase is followed by a split$(j,i)$ phase for some $j \neq i$.
  \item A split$(j,i)$ phase is followed by a sync$(j,j)$ phase.
\end{enumerate}
Therefore, every sync-split$(i,j)$ pair is followed by a sync-split$(j,k)$ pair for some $k \neq j$.
\end{restatable}

The above lemma establishes the key property that allows us to analyze the evolution of the duality gap along the FP dynamic.
We note that this regularity property is missing under the arbitrary tie-breaking condition of Daskalakis and Pan~\cite{daskalakis2014counter}, which enables them to construct a counterexample.
We also note this property 
does not always hold for general $A$, which presents an obstacle to generalizing our result.

Over the course of a sync-split pair, the duality gap increases by at most a constant, as we show in Section~\ref{Sec:PsiSyncSplit}.
Thus, if we can control the length of each sync-split pair, then we can bound the growth of the duality gap.

\subsection{Arguments for upper bound in diagonal case}
\label{Sec:UpperProof}

We now present the key arguments for the fast convergence of Fictitious Play in the diagonal case. In this section, we assume $A$ is a diagonal matrix with $0 < A_{ii} \le A_{\max}$ for all $i \in [n]$, and we let $(x^{(1)}, y^{(1)}) \in \Delta_n \times \Delta_n$ be arbitrary.

Our approach is to invent a potential function (the weight vector) which stays close to the duality gap and can be easily tracked.

\begin{Definition}[Weight vector]
At each time $t \ge 1$, we define the {\em weight vector} $w^{(t)} \in \R^n$ by
\begin{align}\label{Eq:w}
w^{(t)}_i = \frac{\psi(x^{(t)}, y^{(t)})}{A_{ii}} + y^{(t)}_i - x^{(t)}_i ~~ \text{ for }~ i \in [n].
\end{align}
\end{Definition}

By construction, the entries of the weight vector are nonnegative: $w^{(t)}_i \ge 0$ for all $t \ge 1$, $i \in [n]$. Furthermore, 
we can interpret each entry of the weight vector as the sum of the regrets of the $x$ and $y$ players for that action; see Section~\ref{Sec:DetailsDef}.

By Lemma~\ref{Lem:SyncSplit}, we know the FP dynamic proceeds through a sequence of sync-split pairs.
Let $T_1, T_2, \dots$ denote the starting times of the sync-split pairs.
For $s \ge 1$, let the $s^{\text{th}}$ pair be a sync-split$(i_s \to i_{s+1})$ pair (for some $i_{s+1} \neq i_s$) which starts at time $T_s$. Then we can show that the entries of the weight vector are bounded below by the duality gap at the beginning of each sync-split pair.

\begin{lemma}\label{Lem:Increase} 
   In the diagonal case, for all $s \ge 1$:
   \begin{align*}
   w^{(T_s)}_i \ge  \frac{\psi(x^{(T_s)},y^{(T_s)}) - \psi(x^{(T_1)},y^{(T_1)})}{A_{\max}} ~ \text{ for all } i \neq i_s.
   \end{align*}
\end{lemma}

We also show the length of a sync-split pair is bounded below by an entry of the weight vector.

\begin{lemma}\label{Lem:Length}
  In the diagonal case, for all $s \ge 1$:
  \begin{align*}
  T_{s+1}-T_s \ge w^{(T_s)}_{i_{s+1}}.
  \end{align*}
\end{lemma}

Lemma~\ref{Lem:Increase} and Lemma~\ref{Lem:Length} give a recurrence between $T_s$ and $\psi$, which implies the following result.

\begin{lemma}\label{Lem:Diag}
In the diagonal case, for all $s \ge 1$:
\begin{align*}
\psi(x^{(T_s)},y^{(T_s)}) \le \psi(x^{(T_1)},y^{(T_1)}) + 3 A_{\max} \sqrt{T_{s+1} - T_1}.
\end{align*}
\end{lemma}

Lemma~\ref{Lem:Diag} is essentially our main theorem since it already gives the correct dependence of the duality gap on time, but only at the beginning of each sync-split pair.
We can easily bound the beginning and end times to deduce the dependence holds for all time as claimed in Theorem~\ref{Thm:Upper}. We give the full details in Section~\ref{Sec:ThmUpperProof}.

\subsection{Arguments for lower bound in identity case}
\label{Sec:LowerProof}

We now present the arguments for the lower bound in the identity case.

In this section, we assume $A = I_n$ and $(x^{(1)},y^{(1)}) = (e_i,e_j)$ for some $i,j \in [n]$.
By Lemma~\ref{Lem:SyncSplit}, we know the FP dynamic goes through a sequence of sync-split pairs.
As before, let $T_1,T_2,\dots$ denote the starting times of the sync-split pairs.
Let the $s^{\text{th}}$ pair be a sync-split$(i_s \to i_{s+1})$ pair for some $i_{s+1} \neq i_s$.

By construction, the duality gap always has an integral value.
In particular, if the duality gap increases, then it must increase by at least $1$.
We show that after at most $n$ sync-split pairs, the duality gap must increase.
This property crucially uses the tie-breaking~\Cref{as:no_tie}.

\begin{restatable}{lemma}{LemDG}\label{Lem:DG}
In the identity case, for all $s \ge 1$,
\begin{align*}
\psi(x^{(T_{s+n})},y^{(T_{s+n})}) \ge \psi(x^{(T_s)},y^{(T_s)}) + 2.
\end{align*}
\end{restatable}

Next, we show that the lengths of the sync-split pairs cannot be too long, so the starting times are increasing at most quadratically.

\begin{restatable}{lemma}{LemLengthUpper}\label{Lem:LengthUpper}
In the identity case, for all $s \ge 1$,
\begin{align*}
T_{s+1} \le 196 \, s^2.
\end{align*}
\end{restatable}

Combining \Cref{Lem:DG,Lem:LengthUpper} yields the following:

\begin{restatable}{lemma}{LemApple}\label{Lem:Apple}
In the identity case, for all $s \ge n+1$ of the form $s = \ell n + 1$ for some $\ell \ge 1$, we have:
\begin{align*}
\psi(x^{(T_s)},y^{(T_s)}) \ge \frac{\sqrt{T_s}}{7n}.
\end{align*}
\end{restatable}

This is essentially our result, and we extend it to all time in Theorem~\ref{Thm:Lower}. The full details are in Section~\ref{Sec:DetailsLower}.

\section{Proofs for upper bound in diagonal case}
\label{Sec:DetailsUpper}

In this section we assume we are in the diagonal case,
namely when $A$ is a diagonal matrix with entries $A_{ii} > 0$.
Let $A_{\min} = \min_{i \in [n]} A_{ii}$ and $A_{\max} = \max_{i \in [n]} A_{ii}$. We let $(x^{(1)}, y^{(1)}) \in \Delta_n \times \Delta_n$ be arbitrary, and we let $(x^{(t)}, y^{(t)}) \in \R^n \times \R^n$ be the iterates of the Fictitious Play dynamic~\eqref{Eq:FP} under the tie-breaking~\Cref{as:no_tie}.

We start by introducing some helpful new notation in \Cref{Sec:DetailsDef}. We then state some properties of the FP dynamic in the subsequent sections to prove the results claimed in Sections~\ref{Sec:SyncSplit} and~\ref{Sec:UpperProof}. Finally, we prove Theorem~\ref{Thm:Upper} in 
\Cref{Sec:ThmUpperProof}.

\subsection{Definitions}
\label{Sec:DetailsDef}

In this section, we define some helpful new notation. 

At time $t \ge 1$, we define the state variables $p^{(t)}, q^{(t)} \in \R^n$ by
\begin{align*}
p^{(t)} &= A y^{(t)} \\
q^{(t)} &= A x^{(t)}.
\end{align*}
In particular, $p_i^{(t)} = e_i^\top Ay^{(t)}$ and $q_j^{(t)} = (x^{(t)})^\top A e_j$ for all $i,j \in [n]$.
Then we see that $i^{(t)}$ and $j^{(t)}$ from \eqref{Eq:FP} can be defined as follows:
\begin{align*}
i^{(t)} &= \arg\min_{i \in [n]} p_i^{(t)} \\
j^{(t)} &= \arg\max_{j \in [n]} q_j^{(t)}.
\end{align*}
We can think of each entry of $p^{(t)}$ as the loss to the $x$ player for playing a given action against the $y$ player's history, and likewise we can think of $q^{(t)}$ as the vector of payoffs for the $y$ player.

Let
\begin{align*}
p_\ast^{(t)} &= \min_{i\in [n]} p_i^{(t)} = p^{(t)}_{i^{(t)}} \\
q^{\ast (t)} &= \max_{j\in [n]} q_j^{(t)} = q^{(t)}_{j^{(t)}}. 
\end{align*}
We can interpret $p_\ast^{(t)}$ as the loss to the $x$ player of her best action against the $y$ player's history up to round $t$, and we can interpret $q^{\ast (t)}$ analogously.

Then we can write the duality gap at time $t$ as
\begin{align}\label{Eq:Psi}
\psi(x^{(t)},y^{(t)}) = q^{\ast (t)} - p^{(t)}_\ast.
\end{align}

We also define the {\em gap vectors} $u^{(t)}, v^{(t)} \in \R^n$ by
\begin{align*}
\begin{split}
u^{(t)}_i &= p^{(t)}_i - p_\ast^{(t)} \\
v^{(t)}_j &= q^{\ast (t)}- q^{(t)}_j
\end{split}
\end{align*}
for all $i,j \in [n]$. 
The entries of $u^{(t)}$ and $v^{(t)}$ are the regrets of the actions, namely how far they are from being optimal for the $x$ and $y$ players, respectively. 
Note that $u^{(t)}$ and $v^{(t)}$ have nonnegative entries, and each of them has at least one entry equal to $0$ (since $u^{(t)}_{i^{(t)}} = v^{(t)}_{j^{(t)}} = 0$).
Furthermore, observe that we can write the weight vector~\eqref{Eq:w} as the sum of the gap vectors:
\begin{align}\label{Eq:w2}
w^{(t)}_i = \frac{u^{(t)}_i + v^{(t)}_i}{A_{ii}} ~~~ \text{ for all } i \in [n].
\end{align}
In particular, this shows that the entries of the weight vector are always nonnegative.

Since $A$ is diagonal, the FP update~\eqref{Eq:FP} implies the following dynamic on the state variables:
\begin{align}\label{Eq:FPState}
\begin{split}
p^{(t+1)} &= p^{(t)} + A_{j^{(t)} j^{(t)}} \, e_{j^{(t)}} \\
q^{(t+1)} &= q^{(t)} + A_{i^{(t)} i^{(t)}} \, e_{i^{(t)}}.
\end{split}
\end{align}

\subsection{Duality gap never decreases}
\label{Sec:PsiIncreasing}

Recall the type of round $t$ is $(i^{(t)},j^{(t)})$.
Let $\Delta \psi^{(t)}$ denote the change in the duality gap:
\begin{align*}
\Delta \psi^{(t)} =  \psi(x^{(t+1)},y^{(t+1)}) - \psi(x^{(t)},y^{(t)}).
\end{align*}
A basic property of the FP dynamic is that the duality gap never decreases, and it only increases when the type changes. We prove this by carefully tracking $p^{(t)}$ and $q^{(t)}$ across subsequent rounds.

\begin{lemma}\label{Lem:DeltaPsi}
For $t \ge 1$, $0 \le \Delta \psi^{(t)} \le A_{\max}$.
If rounds $t$ and $t+1$ have the same type, then $\Delta \psi^{(t)} = 0$.
\end{lemma}

\begin{proof}
Let the type of round $t$ be $(i,j)$ and the type of round $t+1$ be $(i',j')$.
By definition, this means:  $q^{(t)}_{j} \ge q^{(t)}_{j'}$, $q^{(t+1)}_{j'} \ge q^{(t+1)}_{j}$, and $p^{(t)}_{i} \le p^{(t)}_{i'}$, $p^{(t+1)}_{i'} \le p^{(t+1)}_{i}$.

By definition of the duality gap~\eqref{Eq:Psi} and the FP update rule~\eqref{Eq:FPState}, we have 
\begin{align*}
\psi&(x^{(t+1)},y^{(t+1)}) = q^{(t+1)}_{j'} - p^{(t+1)}_{i'} \\
&=  q^{(t)}_{j'} + A_{j'j'} e_i(j') - ( p^{(t)}_{i'} + A_{i'i'} e_j(i') ) \\
&= \psi(x^{(t)},y^{(t)}) - v^{(t)}_{j'} - u^{(t)}_{i'} + A_{j'j'} e_i(j')  - A_{i'i'} e_j(i').
\end{align*}
Here $e_i(j') = 1$ if $i = j'$, and $e_i(j') = 0$ else.
Therefore, the change in the duality gap is
\begin{align*}
\Delta \psi^{(t)} =  - v^{(t)}_{j'} - u^{(t)}_{i'} + A_{j'j'} e_i(j')  - A_{i'i'} e_j(i').
\end{align*}

First, note that since $u^{(t)}_{i'}, v^{(t)}_{j'} \ge 0$, we have 
$\Delta \psi^{(t)} \le A_{i'i'} e_i(j')  \le A_{\max}$.

Second, from $q^{(t+1)}_{j'} \ge q^{(t+1)}_{j}$ we have
  $q^{(t)}_{j'} + A_{j'j'} e_i(j') \ge q^{(t)}_{j} + A_{jj} e_i(j)$, 
  so $v^{(t)}_{j'} \le A_{j'j'} e_i(j') - A_{jj} e_i(j)$.
  Similarly,  
  from $p^{(t+1)}_{i'} \le p^{(t+1)}_{i}$ we have
  $p^{(t)}_{i'} + A_{i'i'} e_j(i') \le p^{(t)}_{j} + A_{ii} e_j(i)$, 
  so $u^{(t)}_{i'} \le A_{ii} e_j(i) - A_{i'i'} e_j(i')$.
  Since $A_{ii} e_j(i) = A_{jj} e_i(j)$, adding these two inequalities implies $\Delta \psi^{(t)} \ge 0$.

Finally, if $(i',j') = (i,j)$, then we have $v^{(t)}_{j'} = v^{(t)}_{j} = 0$, $u^{(t)}_{i'} = u^{(t)}_{i}  = 0$, and $A_{ii} e_i(j') = A_{jj} e_j(i')$, so $\Delta \psi^{(t)} = 0$.
\end{proof}

\subsection{Proof of Lemma~\ref{Lem:SyncSplit}}

Now we can prove that the FP dynamic alternates through sync and split phases.

\LemSyncSplit*

\begin{proof}
Suppose we are in round $t \ge 1$.

\begin{enumerate}

\item Suppose round $t$ is a sync$(i,i)$ round for some $i \in [n]$, so we are in a sync phase.
This means $q^{\ast (t)} = q^{(t)}_i$ and $p^{(t)}_\ast = p^{(t)}_i$.
The FP update~\eqref{Eq:FP} is $p^{(t+1)} = p^{(t)} + A_{ii} e_i$ and $q^{(t+1)} = q^{(t+r)} + A_{ii} e_i$.
In particular, $p^{(t+1)}_h = p^{(t)}_h$ for $h \neq i$, and $p^{(t+1)}_i = p^{(t)}_i+A_{ii}$.
Similarly,  $q^{(t+1)}_h = q^{(t)}_h$ for $h \neq i$, and $q^{(t+1)}_i = q^{(t)}_i+A_{ii}$.
Observe the maximum of $q$ is still on the $i^{\text{th}}$ entry, and it is unique: $q_\ast^{(t+1)} = q^{(t+1)}_i \ge q^{(t+1)}_h + A_{ii}$, $h \neq i$.
This means the $y$ player does not change action, so $j^{(t+1)} = i$.
Thus, round $t+1$ is either a sync$(i,i)$ round (if $i^{(t+1)} = i$) or a split$(j,i)$ round for some $j \neq i$ (if $i^{(t+1)} = j \neq i$).
Therefore, this sync phase is followed by a split$(j,i)$ phase.

\item Now suppose round $t$ is a split$(j,i)$ round for some $j \neq i$, so we are in a split phase.
This means $q^{\ast (t)} = q^{(t)}_i$ and $p^{(t)}_\ast = p^{(t)}_j$.
The FP update~\eqref{Eq:FP} is $p^{(t+1)} = p^{(t)} + A_{ii} e_i$ and $q^{(t+1)} = q^{(t)} + A_{jj} e_j$.
In particular, $p^{(t+1)}_h = p^{(t)}_h$ for $h \neq i$, and $p^{(t+1)}_i = p^{(t)}_i+A_{ii}$.
Observe the minimum of $p$ in round $t+1$ is still on the $j^{\text{th}}$ entry, and it is still chosen by the lexicographic order in round $t+1$ (because it was chosen in round $t$, and the set of minimizers cannot gain any new entry).
This means the $x$ player does not change action, so $i^{(t+1)} = j$.
On the other hand, because $q$ only changes in the $j^{\text{th}}$ entry, the set of maximizers of $q$ in round $t+1$ either stays the same or it can only gain $j$ as a new maximizer; concretely, $\J^{(t+1)}$ is either $\J^{(t)}$ or $\J^{(t)} \cup \{j\}$.
If $\J^{(t+1)} = \J^{(t)}$, then $j^{(t+1)} = j^{(t)} = i$ since the lexicographic ordering is fixed.
If $\J^{(t+1)} = \J^{(t)} \cup \{j\}$, then $j^{(t+1)}$ is either $i$ or $j$.
Thus, round $t+1$ is either a split$(j,i)$ round or a sync$(j,j)$ round.
Therefore, this split phase is followed by a sync$(j,j)$ phase.
\end{enumerate}
\end{proof}

\subsection{Behavior of weight vector}

In this section, we characterize how the duality gap and $w$ change over the course of sync and split phases. 
We show that the duality gap only increases by at most $A_{\max}$ during each sync and split phase and that each entry of $w$ increases by an amount proportional to the increase in the duality gap. We accomplish this by carefully tracking how $p^{(t)}, q^{(t)}, u^{(t)}$, and $v^{(t)}$ change during and between phases.

\subsubsection{Behavior after a sync phase} 

\begin{lemma}\label{lem:sync}
Suppose rounds $\{t, \dots, t+s-1\}$ form a sync$(i,i)$ phase for some $s \ge 1$, and round $t+s+1$ is a split$(j,i)$ round for some $j \neq i$. 
Let $\epsilon = sA_{ii} - u^{(t)}_j$.
Then
	\begin{enumerate}
		\item $0 \le \epsilon \le A_{ii}$. 
		\item $w^{(t+s)}_\ell = w^{(t-1)}_\ell + A_{\ell\ell}^{-1}\epsilon$ for all $\ell \in [n]$.
		\item $\epsilon = \psi(x^{(t+s)},y^{(t+s)}) - \psi(x^{(t)},y^{(t)})$
	\end{enumerate}
\end{lemma}

\begin{proof}
By assumption, for $0 \le r \le s$ we have $j^{(t+r)} = i$.
This means $q_\ast^{(t+r)} = q_i^{(t+r)}$ for  $0 \le r \le s$.
By assumption, we also have $i^{(t+r)} = i$ for $0 \le r \le s-1$, and $i^{(t+s)} = j$.
This means $p_{\ast}^{(t+r)} = p_i^{(t+r)}$ for  $0 \le r \le s-1$, but  $p_{\ast}^{(t+s)} = p_j^{(t+s)}$.

For  $0 \le r \le s-1$, the FP update is $p^{(t+r+1)} = p^{(t+r)} + A_{ii} e_i$ and $q^{(t+r+1)} = q^{(t+r)} + A_{ii} e_i$.
In particular, $p$ and $q$ only change in the $i^{\text{th}}$ coordinate, which increases by $A_{ii}$ in each round.
Explicitly, we have
$p^{(t+s)} = p^{(t)} + s A_{ii} e_i$
and $q^{(t+s)} = q^{(t)} + s A_{ii} e_i$.

By the properties above, we can deduce the following.
First,
\begin{enumerate}
  \item[(a)] $q^{\ast (t+s)} = q^{(t+s)}_i = q^{(t)}_i + sA_{ii} = q^{\ast (t)} + sA_{ii}$.
  \item[(b)] $v^{(t+s)}_i = v^{(t)}_i = 0$.
  \item[(c)] $v^{(t+s)}_\ell = q^{\ast (t+s)} - q^{(t+s)}_\ell = q^{\ast (t)} + sA_{ii} - q^{(t)}_\ell  = v^{(t)}_\ell + s A_{ii}$ for $\ell \neq i$.
\end{enumerate}
Second, we also have
\begin{enumerate}
  \item[(d)] $p^{(t+s)}_\ast = p^{(t+s)}_j = p^{(t)}_j = p^{(t)}_\ast + u^{(t)}_j =  p^{(t)}_\ast + s A_{ii} - \epsilon$
  \item[(e)] $u^{(t+s)}_i = p^{(t+s)}_i - p^{(t+s)}_\ast = p^{(t)}_i + s A_{ii} - ( p^{(t)}_\ast + s A_{ii} - \epsilon ) = u^{(t)}_i + \epsilon = \epsilon$.
  \item[(f)] $u^{(t+s)}_\ell = p^{(t+s)}_\ell - p^{(t+s)}_\ast = p^{(t)}_\ell - ( p^{(t)}_\ast + s A_{ii} - \epsilon )  = u^{(t)}_\ell - s A_{ii} + \epsilon$ for $\ell \neq i$. \\
  (In particular, $u^{(t+s)}_j = u^{(t)}_j - sA_{ii} + \epsilon = 0$.)
\end{enumerate}
In the above, we have used the definition $\epsilon = s A_{ii} - u_j^{(t)}$.
Let us now prove the claimed properties.
\begin{enumerate}
  \item Since $i^{(t+s)} = j$, we have $p^{(t)}_j = p^{(t+s)}_j \le p^{(t+s)}_i = p^{(t)}_i + sA_{ii} = p^{(t)}_\ast + sA_{ii}$, so $\epsilon = s A_{ii} - u_j^{(t)}\ge 0$.
  
  Since $i^{(t+s-1)} = i$, we have $p^{(t)}_j = p^{(t+s-1)}_j \ge p^{(t+s-1)}_i = p^{(t)}_i + (s-1)A_{ii} = p^{(t)}_\ast + (s-1)A_{ii}$, so $\epsilon = s A_{ii} - u_j^{(t)}\le A_{ii}$.

  \item From (b), (c), (e), and (f), we have: $w^{(t+s)}_\ell = A_{\ell\ell}^{-1}(v^{(t+s)}_\ell + u^{(t+s)}_\ell) = w^{(t)}_\ell + A_{\ell\ell}^{-1}\epsilon$ for all $\ell \in [n]$.
  \item From (a) and (d), we have $\psi(x^{(t+s)},y^{(t+s)}) = q^{\ast (t+s)} - p^{(t+s)}_\ast = q^{\ast (t)} - p^{(t)}_\ast + \epsilon = \psi(x^{(t)},y^{(t)}) + \epsilon$.
\end{enumerate}
\end{proof}

\subsubsection{Behavior after a split phase} 

\begin{lemma}\label{lem:split}
Suppose rounds $\{t, \dots, t+s-1\}$ form a split$(j ,i)$ phase for some $s \ge 1$, $j \neq i$, and round $t+s+1$ is a sync$(j,j)$ round.
Let $\epsilon = sA_{jj} - v^{(t)}_j$.
Then
	\begin{enumerate}
		\item $0 \le \epsilon \le A_{jj}$. 
		\item $w^{(t+s)}_\ell = w^{(t-1)}_\ell + A_{\ell\ell}^{-1}\epsilon$ for $\ell \notin \{i,j\}$,
		and
		$w^{(t+s)}_i = w^{(t)}_i + w^{(t)}_j + (A_{ii}^{-1} + A_{jj}^{-1})\epsilon$, and $w^{(t+s)}_j = 0$.
		\item $\epsilon = \psi(x^{(t+s)},y^{(t+s)}) - \psi(x^{(t)},y^{(t)})$.
	\end{enumerate}
\end{lemma}

\begin{proof}
By assumption, for $0 \le r \le s$ we have $i^{(t+r)} = j$.
This means $p_{\ast}^{(t+r)} = p_j^{(t+r)}$ for  $0 \le r \le s$.
By assumption, we also have $j^{(t+r)} = i$ for $0 \le r \le s-1$, and $j^{(t+s)} = j$.
This means $q_\ast^{(t+r)} = q_i^{(t+r)}$ for  $0 \le r \le s-1$, but  $q_\ast^{(t+s)} = q_j^{(t+s)}$.

For  $0 \le r \le s-1$, the FP update is $p^{(t+r+1)} = p^{(t+r)} + A_{ii} e_i$ and $q^{(t+r+1)} = q^{(t+r)} + A_{jj} e_j$.
In particular, $p$ only changes in the $i^{\text{th}}$ coordinate, while $q$ only changes in the $j^{\text{th}}$ coordinate.
Explicitly,
$p^{(t+s)} = p^{(t)} + s A_{ii} e_i$
and $q^{(t+s)} = q^{(t)} + s A_{jj} e_j$.

By the properties above, we can deduce the following.
First,
\begin{enumerate}
  \item[(a)] $p^{(t+s)}_\ast = p^{(t+s)}_j = p^{(t)}_j = p^{(t)}_\ast$. 
  \item[(b)] $u^{(t+s)}_i = p^{(t+s)}_i - p^{(t+s)}_\ast = p^{(t)}_i + s A_{ii} - p^{(t)}_\ast = u^{(t)}_i + sA_{ii}$.
  \item[(c)] $u^{(t+s)}_\ell = p^{(t+s)}_\ell - p^{(t+s)}_\ast = p^{(t)}_\ell - p^{(t)}_\ast  = u^{(t)}_\ell$ for $\ell \neq i$.
  (In particular, $u^{(t+s)}_j = u^{(t)}_j = 0$.)
\end{enumerate}
Second, we also have
\begin{enumerate}
  \item[(d)] $q^{\ast (t+s)} = q^{(t+s)}_j = q^{(t)}_j + sA_{jj} = q^{\ast (t)} - v^{(t)}_j   + sA_{jj} = q^{\ast (t)} + \epsilon$.
  \item[(e)] $v^{(t+s)}_j = 0$. % q^{\ast (t+s)} - q^{(t+s)}_j = q^{\ast (t)} + \epsilon - =   = 0$.
  \item[(f)] $v^{(t+s)}_\ell = q^{\ast (t+s)} - q^{(t+s)}_\ell = q^{\ast (t)} + \epsilon - q^{(t)}_\ell  = v^{(t)}_\ell + \epsilon$ for $\ell \neq j$.
  (In particular, $v^{(t+s)}_i = v^{(t)}_i + \epsilon = \epsilon$.)
\end{enumerate}
In the above, we have used the definition $\epsilon = s A_{jj} - v_j^{(t)}$.
Let us now prove the claimed properties.
\begin{enumerate}
  \item Since $j^{(t+s)} = j$, we have $q^{\ast (t)} = q^{(t)}_i = q^{(t+s)}_i  \le q^{(t+s)}_j = q^{(t)}_j + s A_{jj}$, so $\epsilon = s A_{jj} - v_j^{(t)}\ge 0$.
  
  Since $j^{(t+s-1)} = i$, 
  $q^{\ast (t)} = q^{(t)}_i  = q^{(t+s-1)}_i \ge q^{(t+s-1)}_j = q^{(t)}_j + (s-1) A_{jj}$, so $\epsilon = s A_{jj} - v_j^{(t)} \le A_{jj}$.

  \item From (b), (c), (e), and (f), we have: $w^{(t+s)}_\ell = A_{\ell\ell}^{-1}(v^{(t+s)}_\ell + u^{(t+s)}_\ell) = w^{(t)}_\ell + A_{\ell\ell}^{-1}\epsilon$ for $\ell \notin\{i,j\}$.
  We also have 
  $w^{(t+s)}_j = 0$.
  Moreover, 
  $w^{(t+s)}_i = A_{ii}^{-1}(v^{(t+s)}_i + u^{(t+s)}_i) = A_{ii}^{-1}( \epsilon + u^{(t)}_i + sA_{ii}) = 
  (A_{ii}^{-1} + A_{jj}^{-1}) \epsilon + w^{(t)}_i + s - A_{jj}^{-1} \epsilon = (A_{ii}^{-1} + A_{jj}^{-1}) \epsilon + w^{(t)}_i + w^{(t)}_j$.
    \item  From (a) and (d), we have $\psi(x^{(t+s)},y^{(t+s)}) = q^{\ast (t+s)} - p^{(t+s)}_\ast = q^{\ast (t)} - p^{(t)}_\ast + \epsilon = \psi(x^{(t)},y^{(t)}) + \epsilon$.

\end{enumerate}
\end{proof}

\subsubsection{Behavior after a sync-split pair}
\label{Sec:PsiSyncSplit}

Over the course of a sync-split pair, the weight vector $w$ changes in a precise way. At the start of the sync-split pair, $w$ has $n-1$ non-zero values. Afterward, at the start of the next sync-split pair, each of these values has increased by an amount proportional to the increase in the duality gap, and the value in the $j^{\text{th}}$ coordinate has moved to the $i^{\text{th}}$ coordinate.

\begin{lemma}\label{lem:mass-cons}
	Suppose rounds $\{t, \dots, t+s-1\}$ form a sync-split$(i\rightarrow j)$ pair for some $i \neq j$.
	Let $\epsilon = \psi(x^{(t+s)},y^{(t+s)}) - \psi(x^{(t)},y^{(t)})$. Then $0 \le \epsilon \le 2A_{\max}$, and we have
	\begin{enumerate}
		\item $w^{(t+s)}_\ell = w^{(t)}_\ell + A^{-1}_{\ell\ell}\epsilon$ for $\ell \neq \{i,j\}$.
		\item $w^{(t+s)}_i = w^{(t)}_j + (A^{-1}_{ii} + A_{jj}^{-1})\epsilon$.
		\item $w^{(t+s)}_j = w^{(t)}_i = 0$.
	\end{enumerate}
\end{lemma}

\begin{proof}
This follows from the characterizations in \Cref{lem:sync,lem:split}.
Let round $t+r-1$ be the last round of the sync phase for this sync-split pair. 
Let $\epsilon_1 = \psi(x^{(t+r)},y^{(t+r)}) - \psi(x^{(t)},y^{(t)})$ and $\epsilon_2 =  \psi(x^{(t+s)},y^{(t+s)}) -  \psi(x^{(t+r)},y^{(t+r)})$, so $\epsilon = \epsilon_1 + \epsilon_2$.
Since $0 \le \epsilon_1,\epsilon_2 \le A_{\max}$, we have $0 \le \epsilon \le 2A_{\max}$.

At the beginning of the sync phase at round $t$, we have $w^{(t)}_i = 0$.
After the sync phase at round $t+r$, by \Cref{lem:sync},
\begin{align*} 
w^{(t+r)}_\ell = w^{(t)}_\ell + A_{\ell\ell}^{-1}\epsilon_1 ~~~~ \text{ for all } ~ \ell \in [n].
\end{align*}
After the split phase at round $t+s$, by \Cref{lem:split}, we have $w^{(t+s)}_j = 0$.
We also have
\begin{align*}
w^{(t+s)}_i &= w^{(t+r)}_i + w^{(t+r)}_j + (A_{ii}^{-1} + A_{jj}^{-1})\epsilon_2 \\
&= w^{(t)}_i + A^{-1}_{ii} \epsilon_1 + w^{(t)}_j + A^{-1}_{jj} \epsilon_1 + (A_{ii}^{-1} + A_{jj}^{-1})\epsilon_2 \\
&= w^{(t)}_j + (A^{-1}_{ii} + A_{jj}^{-1})\epsilon.
\end{align*}
Finally, for $\ell \notin \{i,j\}$, we also have
$w^{(t+s)}_\ell = w^{(t+r)}_\ell + A_{\ell\ell}^{-1} \epsilon_2 
= w^{(t)}_r + A^{-1}_{\ell\ell} \epsilon_1 + A_{\ell\ell}^{-1} \epsilon_2 
= w^{(t)}_r + A^{-1}_{\ell\ell} \epsilon$.
\end{proof}

Lemma~\ref{lem:mass-cons} implies the following bound on the change in weight function in terms of duality gap.
Recall $A_{\min} = \min_{i \in [n]} A_{ii}$ is the minimum and $A_{\max} = \max_{i \in [n]} A_{ii}$ is the maximum entry of $A$.

\begin{corollary}\label{cor:w_increase}
Suppose rounds $\{t, \dots, t+s-1\}$ form a sync-split$(i \rightarrow j)$ pair for some $i \neq j$. 
Let $\epsilon = \psi(x^{(t+s)},y^{(t+s)}) - \psi(x^{(t)},y^{(t)})$. 
Then $0 \le \epsilon \le 2 A_{\max}$, and we have
\begin{enumerate}
  \item $\frac{\epsilon}{A_{\max}} \le w^{(t+s)}_\ell - w^{(t)}_\ell \le \frac{\epsilon}{A_{\min}}$ for $\ell \not\in \{i,j\}$.
  \item $\frac{2\epsilon}{A_{\max}} \le w^{(t+s)}_i - w^{(t)}_j \le \frac{2\epsilon}{A_{\min}}$.
  \item $w^{(t+s)}_j = w^{(t)}_i = 0$.
\end{enumerate}
\end{corollary}

\subsection{The weight vector and the duality gap}
\label{Sec:DetailGap}

In this section we describe a general relation between the weight vector and the duality gap over the course of the FP dynamic.
This follows by inductively applying Corollary~\ref{cor:w_increase}.

Recall by Lemma~\ref{Lem:SyncSplit}, starting from any $(x^{(1)}, y^{(1)}) \in \Delta_n \times \Delta_n$, we know the FP dynamic proceeds through a sequence of sync-split pairs.
Let $T_1,T_2,\dots$ denote the starting times of the sync-split pairs.
Let the $s^{\text{th}}$ pair be a sync-split$(i_s \to i_{s+1})$ pair for some $i_{s+1} \neq i_s$.
Let $\kappa = A_{\max}/A_{\min} \ge 1$ denote the condition number of $A$, where recall $A_{\min} = \min_{i \in [n]} A_{ii}$ and $A_{\max} = \max_{i \in [n]} A_{ii}$.

\subsubsection{Bound on first sync phase}

We first bound $T_1 = \min\{t \ge 1 \colon i^{(t)} = j^{(t)}\}$, the first time we are in a sync phase.

\begin{lemma}\label{Lem:T1}
We have $1 \le T_1 \le \kappa + 2$.
Furthermore, $0 \le w^{(T_1)}_i \le 3\kappa+2$ for all $i \in [n]$.
\end{lemma}
\begin{proof}
If $T_1 = 1$ then we are done,
so assume $T_1 \ge 2$.
The rounds $\{1,\dots,T_1-1\}$ form a split phase, say a split$(j,h)$ phase (and round $T_1$ is a sync$(j,j)$ phase).
By Lemma~\ref{lem:split}, this split phase has length $T_1-1 \le 1 + A_{jj}^{-1} v_j^{(1)} \le 1 + A_{\min}^{-1} v_j^{(1)}$.
Since $v_j^{(1)} \le q_\ast^{(1)} \le A_{\max}$, this implies $T_1 \le \kappa+2$.

From Lemma~\ref{Lem:DeltaPsi} we have $\psi(x^{(T_1)}, y^{(T_1)}) \le \psi(x^{(1)}, y^{(1)}) + A_{\max} \le 2A_{\max}$.
Then from the definition of $w$ in~\eqref{Eq:w}, we have for any $i \in [n]$,
$w^{(T_1)}_i \le A_{ii}^{-1} \psi(x^{(T_1)}, y^{(T_1)}) + y^{(T_1)}_i \le A_{\min}^{-1} (2A_{\max}) + T_1 \le 3\kappa + 2$.
From the expression of $w$ in~\eqref{Eq:w2} 
we also have $w^{(T_1)}_i \ge 0$.
\end{proof}

\subsubsection{Bound between weight vector and duality gap}

We show that the entries of the weight vector are proportional to the duality gap at the beginning of each sync-split pair.
This is a more general form of Lemma~\ref{Lem:Increase}.

\begin{lemma}\label{Lem:wPsi} 
   Under the setting above, for all $s \ge 1$,
   \begin{align*} 
   &\frac{\psi(x^{(T_s)},y^{(T_s)}) - \psi(x^{(T_1)},y^{(T_1)})}{A_{\max}} \\
   &\le w^{(T_s)}_i \le 2\frac{\psi(x^{(T_s)},y^{(T_s)}) - \psi(x^{(T_1)},y^{(T_1)})}{A_{\min}} + 3\kappa+2
   \end{align*}
   for all  $i \neq i_s$. 
\end{lemma}
\begin{proof}
For $s \ge 1$, let $\epsilon_s = \psi(x^{(T_{s+1})},y^{(T_{s+1})}) - \psi(x^{(T_s)},y^{(T_s)})$, and let $C = 3\kappa+2$, so we want to show that
\begin{align}\label{Eq:wPsi}
\frac{1}{A_{\max}} \sum_{r=1}^{s-1} \epsilon_r \le w^{(T_s)}_i \le \frac{2}{A_{\min}} \sum_{r=1}^{s-1} \epsilon_r + C
\end{align}
for all $i \neq i_s$.

We prove this by induction.
The base case $s = 1$ follows from Lemma~\ref{Lem:T1}.
Assume the claim~\eqref{Eq:wPsi} holds for some $s$.
We will show it also holds for $s+1$.

First, for $i \neq \{i_s, i_{s+1}\}$, by Corollary~\ref{cor:w_increase} we have
\begin{align*}
\frac{\epsilon_s}{A_{\max}} \le w^{(T_{s+1})}_i - w^{(T_s)}_i \le \frac{\epsilon_s}{A_{\min}} \le 2\frac{\epsilon_s}{A_{\min}}.
\end{align*}
Combining this with the hypothesis~\eqref{Eq:wPsi} for $i$ at time $T_s$ gives the claim for $i$ at time $T_{s+1}$.

Now for $i = i_{s}$, by Corollary~\ref{cor:w_increase} we also have
\begin{align*}
\frac{\epsilon_s}{A_{\max}} \le 2\frac{\epsilon_s}{A_{\max}} \le w^{(T_{s+1})}_{i_s} - w^{(T_s)}_{i_{s+1}} \le 2\frac{\epsilon_s}{A_{\min}}.
\end{align*}
Combining this with the hypothesis~\eqref{Eq:wPsi} for $i_{s+1}$ at time $T_s$ gives the claim for ${i_s}$ at time $T_{s+1}$.
Thus, we have shown the claim~\eqref{Eq:wPsi} also holds for all $i \neq i_{s+1}$ at time $T_{s+1}$, completing the induction step.
\end{proof}

\subsubsection{Proof of Lemma~\ref{Lem:Increase}}

\begin{proof}[Proof of Lemma~\ref{Lem:Increase}]
This is the lower bound in Lemma~\ref{Lem:wPsi}.
\end{proof}

\subsection{Length of sync-split pairs}

We show the length of a sync-split pair is proportional to an entry of the weight vector.
Here we recall $\kappa = A_{\max}/A_{\min} \ge 1$.

\begin{lemma}\label{Lem:LengthGeneral}
Suppose rounds $\{t,\dots,t+\ell-1\}$ form a sync-split$(i \to j)$ pair for some $i \neq j$.
The length $\ell$ of this sync-split pair is bounded by:
\begin{align*}
w^{(t)}_j \le \ell \le (\kappa+1) w^{(t)}_j + \kappa + 2.
\end{align*}
\end{lemma}

\begin{proof}
Let $t+k$ denote the first split round in this sync-split pair.

By Lemma~\ref{lem:sync}, the sync$(i,i)$ phase $\{t,\dots,t+k-1\}$ has length 
$0 \le k \le 1 + A_{ii}^{-1} u^{(t)}_j \le 1 + \kappa w^{(t)}_j$.

By Lemma~\ref{lem:split}, the split$(j,i)$ phase $\{t+k,\dots,t+\ell-1\}$ has length 
$w^{(t+k)}_j  \le \ell-k \le  1 + w^{(t+k)}_j $
where $w^{(t+k)}_j = A_{jj}^{-1}v^{(t+k)}_j$ since $u^{(t+k)}_j = 0$.
By Lemma~\ref{lem:sync}, we know $w^{(t+k)}_j = w^{(t)}_j + A_{jj}^{-1} \epsilon_1$
where $\epsilon_1 = \psi(x^{(t+k)},y^{(t+k)}) - \psi(x^{(t)},y^{(t)}) \in [0,A_{\max}]$,
so $w^{(t)}_j \le w^{(t+k)}_j \le w^{(t)}_j + \kappa$.
Therefore, we have $w^{(t)}_j \le \ell-k \le w^{(t)}_j + \kappa + 1$.

Combining the two cases above yields 
$w^{(t)}_j \,\le\, \ell = k + (\ell-k) \,\le\, (\kappa+1) w^{(t)}_j + \kappa + 2$.
\end{proof}

\subsubsection{Proof of Lemma~\ref{Lem:Length}}

\begin{proof}[Proof of Lemma~\ref{Lem:Length}]
This is the lower bound in Lemma~\ref{Lem:LengthGeneral} for the sync-split pair $\{T_s,\dots,T_{s+1}-1\}$ with length $T_{s+1}-T_s$.
\end{proof}

\subsection{Bound on duality gap}

We prove the following result on the behavior of the duality gap over the sync-split pairs.
This is a more general form of Lemma~\ref{Lem:Diag}, which is just the upper bound.
As above, let $T_1,T_2,\dots$ denote the starting times of the sync-split pairs.
Let $\kappa = A_{\max}/A_{\min} \ge 1$ where $A_{\min} = \min_{i \in [n]} A_{ii}$ and $A_{\max} = \max_{i \in [n]} A_{ii}$.

\begin{lemma}\label{Lem:DiagGeneral} 
   Under the setting above, for all $s \ge 1$:
   \begin{equation}
   \begin{split}\label{Eq:DiagGeneral}
   \frac{A_{\min}}{2} &\left( \frac{T_{s+1}-T_1}{s (\kappa+1)} - (7\kappa+4) \right)  \\
   &\le\, \psi(x^{(T_s)},y^{(T_s)}) - \psi(x^{(T_1)},y^{(T_1)}) \\
   &\le\, 3A_{\max} \sqrt{T_{s+1}-T_1}.
   \end{split}
   \end{equation}
\end{lemma}
\begin{proof}
For $s \ge 1$, let $\epsilon_s = \psi(x^{(T_{s+1})},y^{(T_{s+1})}) - \psi(x^{(T_s)},y^{(T_s)})$,
so by Lemma~\ref{lem:mass-cons}, $0 \le \epsilon_s \le 2A_{\max}$.
Let $E_s = \sum_{r=1}^s \epsilon_r = \psi(x^{(T_{s+1})},y^{(T_{s+1})}) - \psi(x^{(T_1)},y^{(T_1)})$.

The $s^{\text{th}}$ sync-split pair $\{T_s,\dots,T_{s+1}-1\}$ is a sync-split$(i_s \to i_{s+1})$ pair.
Since $i_s \neq i_{s+1}$, by Lemma~\ref{Lem:wPsi} we have
\begin{align*}
   \frac{1}{A_{\max}} \sum_{r=1}^{s-1} \epsilon_r \,\le\, w^{(T_s)}_{i_{s+1}} \,\le\, \frac{2}{A_{\min}}  \sum_{r=1}^{s-1} \epsilon_r + 3\kappa+2.
\end{align*}
Furthermore, by Lemma~\ref{Lem:LengthGeneral}, the length of this sync-split pair is bounded by
\begin{align*}
w^{(T_s)}_{i_{s+1}} \,\le\, T_{s+1}-T_s \,\le\, (\kappa+1) w^{(T_s)}_{i_{s+1}} + \kappa + 2.
\end{align*}
Combining the two results above yields the following:

\begin{enumerate}
  \item First, combining the lower bounds,
  \begin{align*}
  T_{s+1} - T_1 &= \sum_{r=1}^s (T_{r+1}-T_r) \\
  &\ge \frac{1}{A_{\max}}  \sum_{r=1}^s \sum_{\ell=1}^{r-1} \epsilon_\ell \\
  &= \frac{1}{A_{\max}} \sum_{r=1}^s (s-r+1) \epsilon_r.
  \end{align*}
  Assume for now $E_s \ge 4A_{\max}$.
  Since $0 \le \epsilon_r \le 2A_{\max}$, by the lower bound in Lemma~\ref{Lem:Sum} below we know
  $\sum_{r=1}^s (s-r+1) \epsilon_r \ge \frac{1}{8A_{\max}} E_s^2$.
  Thus, $T_{s+1}-T_1 \ge \frac{1}{8A_{\max}^2} E_s^2$, which implies the desired upper bound in~\eqref{Eq:DiagGeneral}:
  \begin{align*}
  E_s \,\le\, \sqrt{8} A_{\max} \sqrt{T_{s+1}-T_1} \,\le\, 3 A_{\max} \sqrt{T_{s+1}-T_1}.
  \end{align*}
  Now if $E_s < 4A_{\max}$, then we still have the upper bound $E_s < 4A_{\max} \le 3A_{\max} \sqrt{T_{s+1}-T_1}$ since $T_{s+1} \ge T_2 \ge T_1 + 2$ for $s \ge 1$.

  \item Second, combining the upper bounds,
\begin{align*}
&T_{s+1}-T_1 = \sum_{r=1}^s (T_{r+1}-T_r) \\
 &\le (\kappa+1) \sum_{r=1}^s w^{(T_r)}_{i_{r+1}} + s(\kappa + 2) \\
&\le (\kappa+1)  \sum_{r=1}^s \left( \frac{2}{A_{\min}}  \sum_{\ell=1}^{r-1} \epsilon_\ell + 3\kappa+2 \right) + s(\kappa + 2) \\
&= \frac{2(\kappa+1) }{A_{\min}}  \sum_{r=1}^s (s-r+1) \epsilon_r + s (\kappa+1)(3\kappa+2) \\
&~~~~ + s (\kappa + 2).
\end{align*}
Let $E_s = \sum_{r=1}^s \epsilon_r = \psi(x^{(T_{s+1})},y^{(T_{s+1})}) - \psi(x^{(T_1)},y^{(T_1)})$.
Since $0 \le \epsilon_r \le 2A_{\max}$, by the upper bound in Lemma~\ref{Lem:Sum} below we know
$\sum_{r=1}^s (s-r+1) \epsilon_r \le s\left(E_s + 2A_{\max}\right)$.
Therefore,
\begin{align*}
\frac{T_{s+1}-T_1}{s} &\le 
\frac{2(\kappa+1) }{A_{\min}} E_s + 4\kappa(\kappa+1) \\
&~~~~~ + (\kappa+1)(3\kappa+2) + (\kappa + 2).
\end{align*}
Since $4\kappa(\kappa+1) + (\kappa+1)(3\kappa+2) + (\kappa + 2)  = 7\kappa^2 + 10 \kappa + 4 \le (\kappa+1)(7\kappa+4)$, this implies the desired lower bound in~\eqref{Eq:DiagGeneral}.
\end{enumerate}
\end{proof}

\subsubsection{Proof of Lemma~\ref{Lem:Diag}}

\begin{proof}[Proof of Lemma~\ref{Lem:Diag}]
This is the upper bound in Lemma~\ref{Lem:DiagGeneral}.
\end{proof}

\subsection{Proof of Theorem~\ref{Thm:Upper}}
\label{Sec:ThmUpperProof}

We finally prove the upper bound on the fast convergence of the Fictitious Play dynamic.

\ThmUpper*

\begin{proof}[Proof of Theorem~\ref{Thm:Upper}]
Let $T_1,T_2,\dots$ denote the starting times of the sync-split pairs.

Suppose we are at round $t \ge 1$.
Let $s \ge 1$ be such that $T_{s+1} \le t < T_{s+2}$ (if $t < T_2$, see below).
We write the duality gap at time $t$ in terms of at time $T_s$:
\begin{align*}
\psi(x^{(t)},y^{(t)}) \,&\le\, \psi(x^{(T_{s+1})},y^{(T_{s+1})}) + A_{\max} \\
&\le\, \psi(x^{(T_{s})},y^{(T_{s})}) + 3A_{\max}. 
\end{align*}
Furthermore, by Lemma~\ref{Lem:Diag} we know that
\begin{align*}
\psi(x^{(T_s)},y^{(T_s)}) \,&\le\, \psi(x^{(T_1)},y^{(T_1)}) + 3A_{\max} \sqrt{T_{s+1} - T_1} \\
&\le\, 2A_{\max} + 3A_{\max} \sqrt{t}.
\end{align*}
Therefore, at round $t$ we have
$\psi(x^{(t)},y^{(t)}) \,\le\, 5A_{\max} + 3A_{\max}\sqrt{t} \,\le\, 8A_{\max} \sqrt{t}$.
Now if $1 \le t < T_2$, then we also have 
$\psi(x^{(t)},y^{(t)}) \,\le\, \psi(x^{(T_2)},y^{(T_2)})  \le 4A_{\max} < 8A_{\max}\sqrt{t}$.

Thus, we have shown $\psi(x^{(t)},y^{(t)}) \le 8A_{\max}\sqrt{t}$ for all $t \ge 1$.
Then for the scaled iterate,
\begin{align*}
\psi(\hat x^{(t)},\hat y^{(t)})  = \frac{\psi(x^{(t)},y^{(t)})}{t}
\,\le\, \frac{8A_{\max}}{\sqrt{t}}
= O\left(\frac{A_{\max}}{\sqrt{t}}\right)
\end{align*}
as desired.
\end{proof}

\section{Proofs for lower bound in identity case} 
\label{Sec:DetailsLower}

We now give the details for the lower bound in the identity case.

In this section we assume $A = I_n$ is the identity matrix
and we start at the vertices of the simplex, $(x^{(1)}, y^{(1)}) = (e_i, e_j)$ for some $i,j \in \Delta_n$.
Since the updates in the FP dynamic~\eqref{Eq:FP} only involve integer values, all the iterates $(x^{(t)}, y^{(t)})$ also have integer entries.
Since $A = I_n$, the duality gap is also an integer.
In particular, if the duality gap increases, then it must increase by at least $1$.

\begin{remark}
Our lower bound can be generalized, for example to the case when $A$ is a diagonal matrix with rational entries, and the starting points $x^{(1)}, y^{(1)}$ also have rational entries.
Then when the duality gap increases, it must increase by at least $1/K$, where $K$ is the smallest integer such that $KA$, $Kx^{(1)}$, and $Ky^{(1)}$ all have integer entries.
Then the same line of arguments below holds and the bound scales by $1/K$.
For ease of exposition, in this section we present the simple identity case.
\end{remark}

\subsection{Increase in duality gap}

The duality gap can only increase when there is an action switch, namely in between different phases.
Sometimes the duality gap does not increase when there is a tie.
In fact we can characterize explicitly the change in the duality gap and the explicit dependence on the lexicographic tie-breaking.

\subsubsection{Increase after a sync phase}

We recall from Section~\ref{Sec:TieBreak} that $\sigma_x$ and $\sigma_y$ are the permutations generating the lexicographic order in Assumption~\ref{as:no_tie}, and $\I^{(t)}, \J^{(t)}$ are the sets of minimizers and maximizers in each round.

\begin{lemma}\label{lem:large_small}
	Suppose round $t$ is a sync$(i,i)$ round and round $t+1$ is a split$(j,i)$ round for some $j \neq i$. 
	Let $\epsilon = \psi(x^{(t+1)}, y^{(t+1)}) - \psi(x^{(t)},y^{(t)})$. Then:
	\begin{enumerate}
	\item If $\sigma_x(i) > \sigma_x(j)$, then $\epsilon = 0$.
	\item If $\sigma_x(i) < \sigma_x(j)$, then $\epsilon = 1$.
	\end{enumerate}
\end{lemma}
\begin{proof} 
By assumption, 
$p^{(t)}_\ast = p^{(t)}_i$ and $q^{\ast (t)} = q^{(t)}_i$,
while $p^{(t+1)}_\ast = p^{(t+1)}_j$ and $q^{\ast (t+1)} = q^{(t+1)}_i$.
Since round $t$ is a sync$(i,i)$ round,
the FP update is $p^{(t+1)} = p^{(t)} + e_i$ and $q^{(t+1)} = q^{(t)} + e_i$.
In particular, we have $q^{\ast (t+1)}  = q^{\ast (t)}  + 1$.

\begin{enumerate}
  \item  If $\sigma_x(i) > \sigma_x(j)$, then we must have $p^{(t)}_i = p^{(t)}_\ast < p^{(t)}_j$ (otherwise $j \in \I^{(t)}$, and $x$ would have played $j$ at time $t$ due to the tie-break order).
  In particular, since the entries of $p^{(t)}$ are integers, $p^{(t)}_i + 1 \le p^{(t)}_j$.
  From the assumption $p^{(t+1)}_* = p^{(t+1)}_j$, we have $p^{(t)}_j = p^{(t+1)}_j \le p^{(t+1)}_i = p^{(t)}_i + 1$.
  Therefore, we in fact have $p^{(t)}_j = p^{(t)}_i + 1$, and thus $p^{(t+1)}_\ast = p^{(t)}_\ast + 1$.

  This implies $\epsilon = q^{\ast (t+1)}  - p^{(t+1)}_\ast - (q^{\ast (t)} - p^{(t)}_\ast) = 1-1 = 0$, as desired.

  \item If $\sigma_x(i) < \sigma_x(j)$, then we must have $p^{(t+1)}_i > p^{(t+1)}_\ast = p^{(t+1)}_j$ (otherwise $i \in \I^{(t+1)}$, and $x$ would have played $i$ at time $t+1$ due to the tie-break order).
  In particular, since the entries of $p^{(t+1)}$ are integers, $p^{(t+1)}_i \ge p^{(t+1)}_j + 1$.
  Therefore, $p^{(t)}_i = p^{(t+1)}_i  - 1 \ge  p^{(t+1)}_j = p^{(t)}_j$.
  Since $p^{(t)}_\ast = p^{(t)}_i$ by assumption, we also have $p^{(t)}_i \le p^{(t)}_j$, and thus in fact $p^{(t)}_i = p^{(t)}_j$ (this means there is a tie and $i,j \in \I^{(t)}$, but $x$ chooses $i$ since $\sigma_x(i) < \sigma_x(j)$).
  In particular, $p^{(t+1)}_\ast = p^{(t+1)}_j = p^{(t)}_j = p^{(t)}_i = p^{(t)}_\ast$.

  This implies $\epsilon = q^{\ast (t+1)}  - p^{(t+1)}_\ast - (q^{\ast (t)} - p^{(t)}_\ast) = 1-0 = 1$, as desired.

\end{enumerate}
\end{proof}

\subsubsection{Increase after a split phase}

\begin{lemma}\label{lem:large_small_split}
	Suppose round $t$ is a split$(j,i)$ round and round $t+1$ is a sync$(j,j)$ round for some $j \neq i$. 
	Let $\epsilon = \psi(x^{(t+1)}, y^{(t+1)}) - \psi(x^{(t)},y^{(t)})$. Then:
	\begin{enumerate}
	\item If $\sigma_y(i) > \sigma_y(j)$, then $\epsilon = 0$.
	\item If $\sigma_y(i) < \sigma_y(j)$, then $\epsilon = 1$.
	\end{enumerate}
\end{lemma}
\begin{proof} 
By assumption, 
$p^{(t)}_\ast = p^{(t)}_j$ and $q^{\ast (t)} = q^{(t)}_i$, while $p^{(t+1)}_\ast = p^{(t+1)}_j$ and $q^{\ast (t+1)} = q^{(t+1)}_j$.
Since round $t$ is a split$(j,i)$ round,
the FP update is $p^{(t+1)} = p^{(t)} + e_i$ and $q^{(t+1)} = q^{(t)} + e_j$.
In particular, we have $p^{(t+1)}_\ast  = p^{(t)}_\ast$.

\begin{enumerate}
  \item  If $\sigma_y(i) > \sigma_y(j)$, then we must have $q^{(t)}_i = q^{\ast (t)} > q^{(t)}_j$ (otherwise $j \in \J^{(t)}$, and $y$ would have played $j$ at time $t$ due to the tie-break order).
  In particular, since the entries of $q^{(t)}$ are integers, $q^{(t)}_i \ge q^{(t)}_j + 1$.
  From the assumption $q^{\ast(t+1)} = q^{(t+1)}_j$, we also have $q^{(t)}_j + 1 = q^{(t+1)}_j \ge q^{(t+1)}_i = q^{(t)}_i$.
  Therefore, we in fact have $q^{(t)}_i = q^{(t)}_j + 1$, and thus $q^{\ast (t+1)} = q^{(t+1)}_j = q^{(t)}_j + 1 = q^{(t)}_i = q^{\ast (t)}$.
  
  This implies $\epsilon = q^{\ast (t+1)}  - p^{(t+1)}_\ast - (q^{\ast (t)} - p^{(t)}_\ast) = 0-0 = 0$, as desired.

  \item If $\sigma_y(i) < \sigma_y(j)$, then we must have $q^{(t+1)}_i < q_\ast^{(t+1)} = q^{(t+1)}_j$ (otherwise $i \in \J^{(t+1)}$, and $y$ would have played $i$ at time $t+1$ due to the tie-break order).
  In particular, since the entries of $q^{(t+1)}$ are integers, $q^{(t+1)}_i + 1 \le q^{(t+1)}_j$.
  Therefore, $q^{(t)}_i = q^{(t+1)}_i  \le  q^{(t+1)}_j-1 = q^{(t)}_j$.

  Since $q^{\ast (t)} = q^{(t)}_i$ by assumption, we also have $q^{(t)}_i \ge q^{(t)}_j$, and thus in fact $q^{(t)}_i = q^{(t)}_j$ (this means there is a tie and $i,j \in \J^{(t)}$, but $y$ chooses $i$ since $\sigma_y(i) < \sigma_y(j)$).
  In particular, $q^{\ast (t+1)} = q^{(t+1)}_j = q^{(t)}_j + 1 = q^{(t)}_i + 1 = q^{\ast (t)} + 1$.
  
  This implies $\epsilon = q^{\ast (t+1)}  - p^{(t+1)}_\ast - (q^{\ast (t)} - p^{(t)}_\ast) = 1-0 = 1$, as desired.

\end{enumerate}
\end{proof}

\subsubsection{Proof of Lemma~\ref{Lem:DG}}

We now prove that the duality gap must strictly increase over any sequence of $n$ sync-split pairs.
Here recall $T_1, T_2, \dots$ are the starting times of the sync-split pairs in the FP dynamic.
Let the $s^{\text{th}}$ sync-split pair be a sync-split$(i_s \to i_{s+1})$ pair.

\LemDG*

\begin{proof} 
From round $T_{s}$ to round $T_{s+n}$, there are $n$ 
transitions from sync$(i_r,i_r)$ phase to split$(i_{r+1},i_r)$ phase, for $s \le r \le s+n-1$.
By Lemma~\ref{lem:large_small}, in each of these transitions, the duality gap stays the same if $\sigma_x$ is decreasing, which can happen at most $n-1$ consecutive times.
Since there are $n$ transitions, the duality gap must increase at least once, and it must increase by at least $1$.

Similarly, from round $T_{s}$ to round $T_{s+n}$, there are $n$ 
transitions from split$(i_{r+1},i_r)$ phase to sync$(i_{r+1},i_{r+1})$ phase, for $s \le r \le s+n-1$.
By Lemma~\ref{lem:large_small_split}, in each of these transitions the duality gap stays the same if $\sigma_y$ is decreasing, which can happen at most $n-1$ consecutive times.
Since there are $n$ transitions, the duality gap must increase at least once, and it must increase by at least $1$.
Combining the two contributions above, we conclude that from round $T_{s}$ to round $T_{s+n}$, the duality gap must increase by at least $2$.
\end{proof}

\subsection{Proof of Lemma~\ref{Lem:LengthUpper}}

We show the starting times of the sync-split pairs are increasing at most quadratically.

\LemLengthUpper*

\begin{proof}
Since we are in the identity case, $\kappa = 1$.
The left and right sides of the bound~\eqref{Eq:DiagGeneral} imply the following inequality for $\phi = \sqrt{T_{s+1}-T_1}$:
\begin{align*}
\phi^2 - 12 s \phi - 22 s \le 0.
\end{align*}
This quadratic inequality implies for $\phi \ge 0$:
\begin{align*}
\phi &\le 6s + \sqrt{36 s^2 + 22s} = 6s \left(1 + \sqrt{1 + \frac{11}{18s}} \right)  \\
&\le\, 6s \left(2 + \frac{11}{36s}\right)
\,\le\, 12s + 2
\end{align*}
where we have used the inequality $\sqrt{1+x} \le 1 + x/2$.
Since $T_1 \le 3$ from Lemma~\ref{Lem:T1}, this implies
\begin{align*}
T_{s+1} = T_1 + \phi^2 
\,\le\, 3 + (12s+2)^2
\,\le\, 196 s^2
\end{align*}
where the last inequality holds for $s \ge 2$.

If $s = 1$, by Lemma~\ref{Lem:LengthGeneral} we know $T_2 < T_1 + 10 < 196$,
so the bound still holds.
\end{proof}

\subsection{Proof of Lemma~\ref{Lem:Apple}}

Then we can prove that the duality gap at the beginning of every $n$ sync-split pairs is bounded below by the square root of the starting time.

\LemApple*

\begin{proof}
Let $s = \ell n + 1$ for some $\ell = \frac{(s-1)}{n} \ge 1$.
By iterating~Lemma~\ref{Lem:DG} for $\ell$ times, we get
\begin{align*}
\psi(x^{(T_{s})},y^{(T_{s})}) \,&\ge\, \psi(x^{(T_{1})},y^{(T_{1})}) + 2 \ell \\
\,&\ge\, 0 + \frac{2}{n} (s-1)
\,\ge\, \frac{2}{n} \frac{\sqrt{T_s}}{14}
\,=\, \frac{\sqrt{T_s}}{7n}
\end{align*}
where in the last inequality we have used Lemma~\ref{Lem:LengthUpper}.
\end{proof}

\subsection{Proof of Theorem~\ref{Thm:Lower}}

We now prove the lower bound in Theorem~\ref{Thm:Lower} by extending the result in Lemma~\ref{Lem:Apple} to all time $t \ge 1$.

\ThmLower*

\begin{proof}
Suppose we are at round $t \ge 1$.
Let $\ell \ge 0$ be such that for $s = \ell n + 1$, $T_{s} \le t < T_{s+n}$.
From round $t$ to round $T_{s+n}$ there are at most $n$ sync-split pairs, during each of which $\psi$ can increase by at most $2$, so
$\psi(x^{(t)},y^{(t)}) \,\ge\, \psi(x^{(T_{s+n})},y^{(T_{s+n})}) - 2n$.
Therefore, by Lemma~\ref{Lem:Apple}, 
\begin{align*}
\psi(x^{(t)},y^{(t)}) 
 \,\ge\, \frac{\sqrt{T_{s+n}}}{7n} - 2n
 \,\ge\, \frac{\sqrt{t}}{7n} - 2n.
\end{align*}
Thus, for the scaled iterate,
\begin{align*}
\psi(\hat x^{(t)},\hat y^{(t)})  &= \frac{\psi(x^{(t)},y^{(t)})}{t}
\,\ge\, \frac{1}{t}\left(\frac{\sqrt{t}}{7n} - 2n\right) \\
&= \Omega\left(\frac{1}{n\sqrt{t}}\right).
\end{align*}
\end{proof}

\section{Discussion}
\label{Sec:Disc}

In this paper we have demonstrated a $\Theta(t^{-1/2})$ convergence rate for FP with lexicographic tie-breaking for diagonal payoff matrices. Our work leaves several possible directions to explore.

One immediate question is whether we can extend the fast convergence result from the diagonal case to more general classes of matrices. For general matrices, the sync-split structure of the FP dynamic (Lemma~\ref{Lem:SyncSplit}) no longer holds, as there can be multiple consecutive split phases.
Moreover, the potential function (the weight vector) is no longer proportional to the duality gap. 
Despite these, we observe that the $O(t^{-1/2})$ convergence rate seems to hold empirically (see \Cref{fig:empirical_karlin}), 
suggesting that Karlin's conjecture holds more generally.

\begin{figure}[ht!]
	\centering
	\begin{subfigure}[t]{0.48\textwidth}
		\centering
		\includegraphics[width=.8\textwidth]{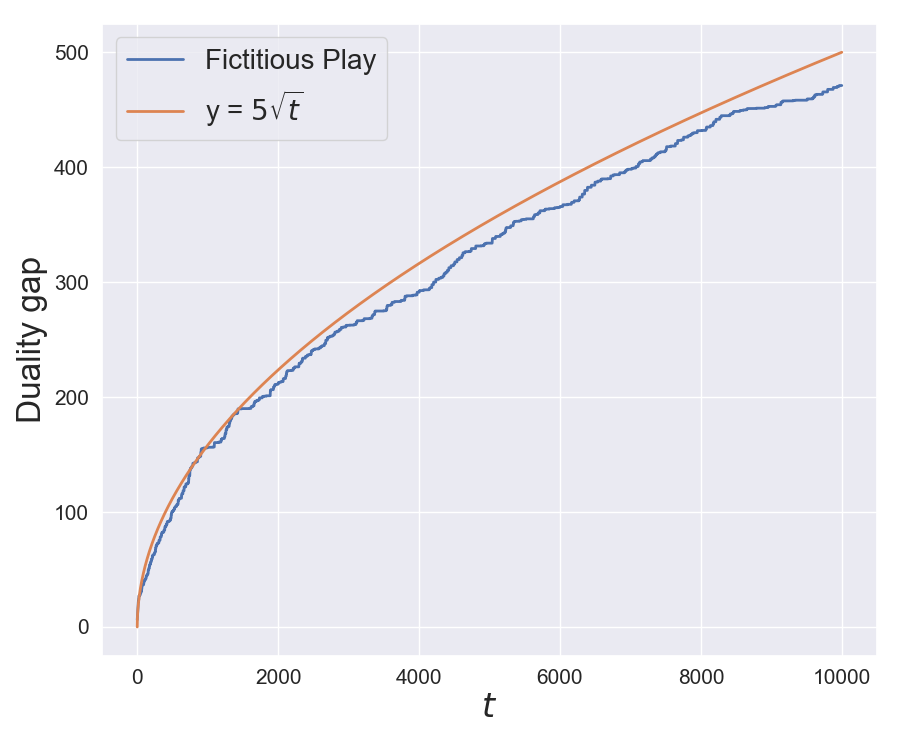}
		\caption{Max duality gap vs. iteration $t$}
	\end{subfigure}
	~ 
	\begin{subfigure}[t]{0.48\textwidth}
		\centering
		\includegraphics[width=.8\textwidth]{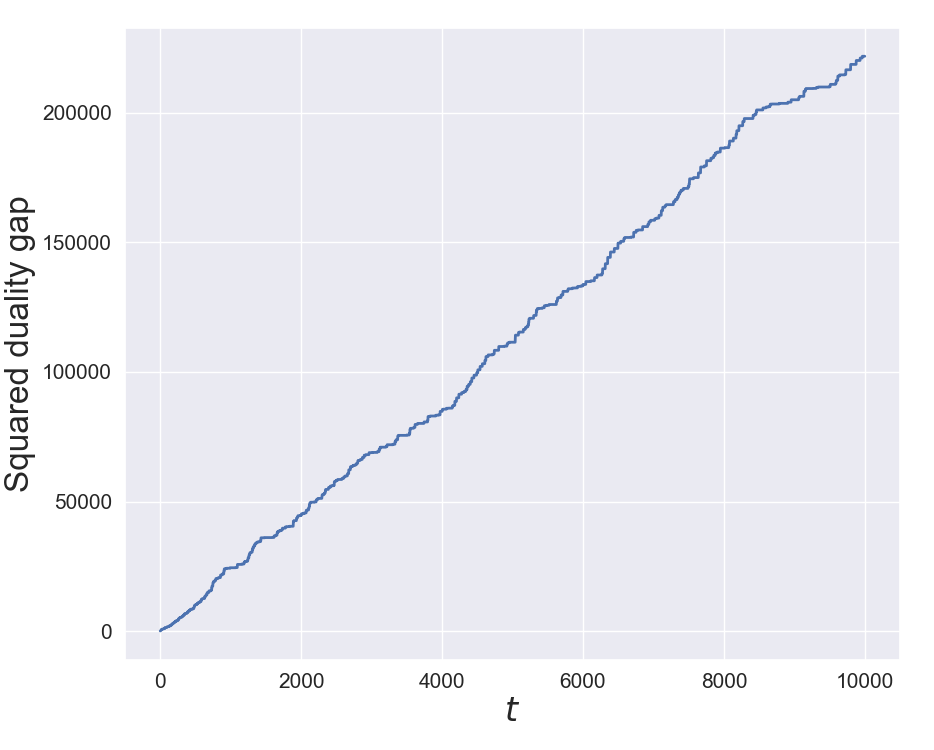}
		\caption{Squared duality gap vs. iteration $t$}
	\end{subfigure}
	\caption{These plots show the maximum duality gap of Fictitious Play at each iteration, where the maximum is over 100 runs of FP, each on a different random $10\times 10$ Gaussian payoff matrix. We see that the duality gap of FP is bounded by $O(\sqrt{t})$ after $t$ iterations. \label{fig:empirical_karlin}
	}
	
\end{figure}

As a first step, we can try to consider a more restricted class of payoff matrices. For example, when the payoff matrix is {\em doubly stochastic}, which means each row and column sums to $1$, one can show an analog of Lemma~\ref{lem:mass-cons}, except that the relations do not hold entrywise, but they hold for the sum of the entries of the weight vector. However, one still lacks an analog for \Cref{Lem:Length} to show that the weight vector relates to the phase length.

Our analysis also shows promise for proving the convergence of related dynamics. As described in \Cref{app:discretization}, FP can be viewed as a forward Euler discretization of continuous-time FP. Two related algorithms are Alternating Fictitious Play (AFP) and Optimistic Fictitious Play (OFP), which correspond to using different methods to discretize the continuous-time flow. 

Alternating Fictitious Play is defined as follows:
\begin{align*}
\begin{split}
x^{(t+1)} &= x^{(t)} + e_{ i^{(t)} }  \\ 
y^{(t+1)} &= y^{(t)} + e_{ j^{(t+1)} }.  
\end{split}
\end{align*}
In AFP, the players take turns to play best response to the opponent's history, rather than best-responding simultaneously as in standard FP.
This method was also described in Brown's original paper on FP \cite{brown1949some}, and Robinson's $O(t^{-\frac{1}{2n-2}})$ convergence result applies to AFP as well \cite{robinson1951iterative}. More recently, alternating versions of dynamics have been shown to have favorable properties compared to their simultaneous counterparts \cite{gidel2019negative, bailey2020finite}. 
Similar to FP, AFP seems to achieve a $O(1/\sqrt{t})$ convergence rate empirically.

Optimistic Fictitious Play is defined as follows:
\begin{align}
\begin{split}
x^{(t+1)} &= x^{(t)} + 2e_{ i^{(t)} } - e_{ i^{(t-1)} }\\
y^{(t+1)} &= y^{(t)} + 2e_{ j^{(t)} } - e_{ j^{(t-1)} }
\end{split}
\end{align}
In OFP, each player predicts that the opponent will repeat their last action one more time, and then plays a best response to the opponent's history plus the predicted action. In doing so, OFP attempts to approximate the Be-The-Leader dynamic, whose duality gap is bounded by $O(1/t)$ \cite{kalai2005efficient}. A recent line of work has showed that optimistic algorithms can improve convergence rates in game settings \cite{rakhlin2013optimization, syrgkanis2015fast, abernethy2018faster, wang2018acceleration}. For zero-sum games, OFP appears to have a $O(1/t)$ bound on its duality gap empirically.

We believe that our analysis tools, such as the weight vector, could be carried over to AFP and OFP.
We conjecture that for diagonal payoff matrices, AFP and OFP have convergence rates $O(1/\sqrt{t})$ and $O(1/t)$, respectively, as this appears to hold empirically.

\appendix

\section{Proof of Lemma~\ref{Lem:Minimax}}
\label{Sec:MinimaxProof}

\begin{proof}[Proof of Lemma~\ref{Lem:Minimax}]
From the definition of the duality gap~\eqref{Eq:DG}, for any $(x,y) \in \Delta_n \times \Delta_n$ we have
\begin{align*}
\psi(x,y) \ge x^\top Ay - x^\top Ay = 0.
\end{align*}

Now if $(x^\ast,y^\ast) \in \Delta_n \times \Delta_n$ is such that $\psi(x^\ast,y^\ast) = 0$, then 
\begin{align*}
\max_{ y \in \Delta_n} \, (x^\ast)^\top A  y  = (x^\ast)^\top Ay^\ast
= \min_{ x \in \Delta_n} \,  x^\top Ay^\ast.
\end{align*}
Therefore, for all $(x,y) \in \Delta_n \times \Delta_n$, we have
$(x^\ast)^\top A y \le (x^\ast)^\top A y^\ast \le x^\top A y^\ast$,
which means $(x^\ast,y^\ast)$ is a minimax point.

Conversely, if $(x^\ast,y^\ast)$ is a minimax point, which means
$(x^\ast)^\top A y \le (x^\ast)^\top A y^\ast \le x^\top A y^\ast$  for all $(x,y) \in \Delta_n \times \Delta_n$,
then we have $\max_{ y \in \Delta_n} \, (x^\ast)^\top A  y = (x^\ast)^\top Ay^\ast =  \min_{ x \in \Delta_n} \,  x^\top Ay^\ast$, 
and therefore $\psi(x^\ast,y^\ast) = 0$.
\end{proof}

\section{A geometric view of Fictitious Play}
\label{App:Geom}

Let $\Z = \Delta_n \times \Delta_n \subset \R^{2n}$.
We write $(x,y) \in \Delta_n \times \Delta_n$ as $z = (x,y) \in \Z$.

We observe that we can write the duality gap $\psi(z) = \psi(x,y)$ from~\eqref{Eq:DG} in terms of the support function of $\Z$:
\begin{align*}
\psi(z) &= \max_{\tilde y \in \Delta_n} \, x^\top A \tilde y - \min_{\tilde x \in \Delta_n} \, \tilde x^\top Ay \\
&= \max_{\tilde z = (\tilde x, \tilde y) \in \Z} \begin{pmatrix}
\tilde x \\ \tilde y
\end{pmatrix}^\top
\begin{pmatrix}
-Ay \\ A^\top x
\end{pmatrix} \\
&= \max_{\tilde z \in \Z} \tilde z^\top Sz \\
&= \phi_\Z(Sz)
\end{align*}
where we have defined the skew-symmetric matrix 
\begin{align*}
S = \begin{pmatrix}
0 & -A \\ A^\top & 0
\end{pmatrix} \in \R^{2n \times 2n}.
\end{align*}

Here $\phi_\Z \colon \R^{2n} \to \R$ is the {\em support function} of $\Z$, which is defined by $\phi_\Z(\theta) = \max_{z \in \Z} \, \theta^\top z$.
We recall the subgradient set\footnote{The subgradient set of a convex function $\phi \colon \R^m \to \R$ at $\theta \in \R^m$ is the set $\partial \phi(\theta) = \{g \in \R^m \colon \phi(x) \ge \phi(\theta) + g^\top (x-\theta) ~ \text{ for all } x \in \R^m\}$} of the support function is the set of maximizers:
\begin{align*}
\partial \phi_\Z(\theta) = \arg\max_{z \in \Z} \, \theta^\top z.
\end{align*}
Therefore, the Fictitious Play dynamic for $z^{(t)} = (x^{(t)}, y^{(t)})$ is an instance of the update rule\footnote{In general, we can choose any element $\omega^{(t)}$ from the subgradient set $\partial \phi_\Z(Sz^{(t)})$ to make the update.
In our formulation of FP~\eqref{Eq:FP}, we choose a particular extreme point $\omega^{(t)} = (e_{i^{(t)}}, e_{j^{(t)}})$ based on the lexicographic ordering.}
\begin{align}\label{Eq:FPGeom}
z^{(t+1)} = z^{(t)} + \omega^{(t)}, ~~~ \omega^{(t)} \in \partial \phi_\Z(Sz^{(t)}).
\end{align}

This geometric view makes it clear that Fictitious Play increases the duality gap.
Indeed, since the vector $S^\top \omega^{(t)}$ is in the subgradient set $\partial \psi(z^{(t)}) = S^\top \partial \phi_\Z(Sz^{(t)})$, this means
\begin{align*}
\psi(z^{(t+1)}) &\ge \psi(z^{(t)}) + (S^\top \omega^{(t)})^\top (z^{(t+1)}-z^{(t)}) \\
&\ge \psi(z^{(t)}) + (\omega^{(t)})^\top S \omega^{(t)} \\
&= \psi(z^{(t)}) .
\end{align*}
In the last equality above we have used the fact $S$ is skew-symmetric ($S^\top = -S$), so the quadratic form defined by $S$ is equal to $0$.

\subsection{Behavior in continuous time}

The FP update~\eqref{Eq:FPGeom} is the $\eta = 1$ case of a discrete-time algorithm $z^{(t+1)} = z^{(t)} + \eta \, \partial \phi_\Z(Sz^{(t)})$, for $t = 1,2,\dots$.
As $\eta \to 0$, this algorithm converges to the continuous-time dynamic $Z(t)$, $t \ge 0$, given by
\begin{align}\label{Eq:Flow}
\dot Z(t) = \partial \phi_\Z(SZ^{(t)}).
\end{align}
If $S$ is invertible, then we can write the above as a skew-gradient flow: 
$\dot Z(t) = (S^\top)^{-1} \partial \psi(Z(t))$.
Since $S$ is skew-symmetric, this flow preserves the duality gap:
\begin{align*}
\frac{d}{dt} \psi(Z(t)) &= \partial \psi(Z(t))^\top \dot Z(t) \\
&= \partial \phi_\Z(Z(t))^\top S \partial \phi_\Z(Z(t)) = 0.
\end{align*}
Therefore, for the scaled iterate $\hat Z(t) = \frac{1}{t} Z(t)$, the duality gap decreases at a $\Theta(t^{-1})$ rate:
\begin{align*}
\psi(\hat Z(t)) = \frac{\psi(Z(t))}{t} = \frac{\psi(Z_1)}{t} = \Theta(t^{-1}).
\end{align*}

\subsubsection{Discretization methods}\label{app:discretization}

We can view the FP update~\eqref{Eq:FPGeom} as a forward discretization (also known as explicit Euler method) of the continuous-time dynamic~\eqref{Eq:Flow}.
In discrete time, this forward discretization does not preserve the duality gap due to the discretization error, and by convexity the duality gap is always increasing.

We can consider other possible algorithms by using other discretization methods.
For example, the backward discretization (or the implicit Euler method) is guaranteed to decrease the duality gap by convexity.
This implicit method is not necessarily implementable as a strategy, but there are approximations to it, for example via the optimistic method.
Another possible discretization is the symplectic Euler method, which should conserve the duality gap better.
This corresponds to the alternating version of Fictitious Play, as we mention in Section~\ref{Sec:Disc}.

\section{A helper lemma}
\label{Sec:Helper}

In the proof of Lemma~\ref{Lem:DiagGeneral} above we use the following result.
Here for $x \in \R$, let $\lfloor x \rfloor$ denote the floor of $x$, which is the largest integer less than or equal to $x$.
In particular, $x-1 \le \lfloor x \rfloor \le x$.

\begin{lemma}\label{Lem:Sum}
Let $0 \le \epsilon_1,\dots,\epsilon_s \le \epsilon_{\max}$ for some $s \ge 1$, and let $E = \sum_{r=1}^s \epsilon_r$.
Then:
\begin{enumerate}
  \item $\sum_{r=1}^s (s-r+1) \epsilon_r \,\le\, s(E+ \epsilon_{\max})$.
  \item If $E \ge 2\epsilon_{\max}$, then $\sum_{r=1}^s (s-r+1) \epsilon_r \,\ge\, \frac{1}{4\epsilon_{\max}} E^2$.
\end{enumerate}
\end{lemma}
\begin{proof}
If $E = 0$, then all $\epsilon_r = 0$ and we are done.
Now assume $E > 0$.

For fixed $E = \sum_{r=1}^s \epsilon_r$, the maximum of $\sum_{r=1}^s (s-r+1) \epsilon_r$ is achieved when $\epsilon_1 = \cdots = \epsilon_{m} = \epsilon_{\max}$ where $m = \lfloor E/\epsilon_{\max} \rfloor$, $\epsilon_{m+1} = E - m\epsilon_{\max}$, and $\epsilon_{m+2} = \cdots = \epsilon_s = 0$.
This gives
\begin{align*}
\sum_{r=1}^s (s-r+1) \epsilon_r &\le \sum_{r=1}^{m+1} (s-r+1) \epsilon_{\max} \le s(m+1) \epsilon_{\max} \\
&\le s \left(\frac{E}{\epsilon_{\max}}+1\right) \epsilon_{\max}
=  s(E + \epsilon_{\max}).
\end{align*}

Similarly, for fixed $E = \sum_{r=1}^s \epsilon_r$, the minimum of $\sum_{r=1}^s (s-r+1) \epsilon_r$ is achieved when $\epsilon_s = \cdots = \epsilon_{s-m+1} = \epsilon_{\max}$ where $m = \lfloor E/\epsilon_{\max} \rfloor$, $\epsilon_{s-m} = E - m\epsilon_{\max}$, and $\epsilon_{s-m-1} = \cdots = \epsilon_1 = 0$.
This gives
\begin{align*}
\sum_{r=1}^s (s-r+1) \epsilon_r &\ge \sum_{r=s-m+1}^s (s-r+1) \epsilon_{\max} \\
&=  \frac{m(m+1)}{2} \epsilon_{\max} \\
&\ge \left(\frac{E}{\epsilon_{\max}}-1\right) \frac{E}{\epsilon_{\max}} \frac{\epsilon_{\max}}{2}
\ge \frac{E^2}{4\epsilon_{\max}}
\end{align*}
where the last inequality holds if $E \ge 2\epsilon_{\max}$.
\end{proof}

\section*{Acknowledgements}
The authors thank Georgios Piliouras for several valuable and insightful discussions.

\bibliographystyle{plain}
\bibliography{FPBib}

\end{document}